\documentclass{LMCS}
\usepackage[utf8]{inputenc}
\usepackage[english]{babel}
\usepackage{amsmath,amssymb}
\usepackage{graphics}
\usepackage{subfig}
\usepackage{enumerate}
\usepackage{mathtools}
\usepackage{algorithmic}
\usepackage{algorithm}
\usepackage{listings}
\usepackage{linlog}
\usepackage{url}
\usepackage{hyperref}

\usepackage{tikz}
\usepackage[nofancy]{tikz-inet}
\usetikzlibrary{calc}

\usetikzlibrary{arrows}

\newcommand{\Orb}{\mbox{Orb}}
\newcommand{\Orbs}{\mbox{Orbs}}

\newcommand{\Ex}{\mathsf{Ex}}

\newcommand{\permplus}{+}
\newcommand{\dom}{\mathrm{dom}}
\newcommand{\codom}{\mathrm{codom}}
\newcommand{\pal}{\mbox{pal}}

\newcommand{\restr}[2]{ {#1}{\upharpoonright_{#2}} }
\newcommand{\frestr}[3]{{{#1}\!\restriction}^{#2}_{#3}}

\newcommand{\perm}[1]{\mathfrak{S}({#1})}

\newcommand{\cat}[1]{\mathsf{#1}}

\newcommand{\st}{~,~}
\newcommand{\set}[2]{\{#1~|~#2\}}
\newcommand{\NN}{\mathbb{N}}
\newcommand{\glabel}{\mathcal{L}}
\newcommand{\compsymb}{\leftrightsquigarrow}
\newcommand{\piledperp}[1]{\stackrel{#1}{\compsymb}}
\newcommand{\piledpperp}[1]{\piledperp{#1}_0}
\newcommand{\cexec}[1]{\piledpperp{#1}}
\newcommand{\cexece}[1]{\piledperp{#1}}
\newcommand{\glue}[1]{\piledperp{#1}}
\newcommand{\cglue}[0]{\compsymb}
\newcommand{\ACglue}[0]{\leftrightarrow}
\newcommand{\cord}[2]{\rho(#1,#2)}
\newcommand{\context}[2]{{#1}^{#2}}
\newcommand{\icontains}{\subset}
\newcommand{\ccglue}[4]{\context{#1}{#2} \cglue \context{#3}{#4}}
\newcommand{\point}[1]{\stackrel{\bullet}{#1}}
\newcommand{\exc}[1]{\Ex(#1)}
\newcommand{\exequiv}{\stackrel{\leftrightsquigarrow}{\sim}}
\newcommand{\dglue}{\permplus}

\newcommand{\isoarrow}{\twoheadrightarrow}

\newcommand{\cate}[1]{\mathsf{#1}}
\newcommand{\IN}{\cate{IN}}

\newcommand{\xinjarrow}{\xhookrightarrow}
\newcommand{\xisoarrow}[1]{\xrightarrow[\sim]{#1}}
\newcommand{\bij}[1]{\widetilde{#1}}

\newenvironment{ctp}{\begin{center}\begin{tikzpicture}}
{\end{tikzpicture}\end{center}}

\newenvironment{itp}{\begin{tikzpicture}[baseline=3pt]}
    {\end{tikzpicture}}

\newcommand{\includetikzinline}[1]{%
    \raisebox{-.4\height}{\includegraphics{#1.pdf}}
}

\newcommand{\includetikzcentered}[1]{%
\begin{center}\includegraphics{#1.pdf}\end{center}
}

\newcommand{\includetikz}[1]{%
    \includegraphics{#1.pdf}
}

\def\doi{6 (4:6) 2010}
\lmcsheading%
{\doi}
{1--27}
{}
{}
{Nov.~\phantom06, 2009}
{Dec.~12, 2010}
{}

\begin{document}
\title{An Explicit Framework for Interaction Nets}
\author[M.~de Falco]{Marc de Falco}
\address{e-on software, 68 avenue Parmentier, 75011 Paris}
\email{marc@de-falco.fr}
\thanks{This work is supported by the French ANR project \textbf{CHoCo}
(ANR-07-BLAN-0324)}

\maketitle
\begin{abstract}
Interaction nets are a graphical formalism inspired by Linear
Logic proof-nets often used for studying higher order rewriting
e.g. $\beta$-reduction. Traditional presentations of interaction nets
are based on graph theory and rely on elementary  properties of graph
theory. We give here a more explicit presentation based on notions borrowed 
from  Girard's \emph{Geometry of Interaction}: interaction nets
are presented as partial permutations and a composition of nets, 
the \emph{gluing}, is derived from the execution formula.
We then define contexts and reduction as the context closure of
rules. 
We prove strong confluence of the reduction within our framework
and show how interaction nets can be viewed as the quotient of some 
generalized proof-nets.
\end{abstract}

\section{Introduction}
Interaction nets were introduced by Yves Lafont in~\cite{Lafont90} as a way
to extract a model of computation from the well-behaved proof-nets
of multiplicative linear logic. They have since been widely used as a
formalism for the implementation of reduction strategies for
the $\lambda$-calculus, providing a pictorial~\footnote{By putting
a visual emphasis on occurrences of a variable, interaction nets allow
a formal reasoning while not being as cumbersome as indices.}
 way to do explicit substitution~\cite{Mackie98a}\cite{Mackie98b}\cite{Lippi03}
and implement optimal reduction~\cite{AbadiGonthierLevy92b}.

Interaction nets are easy to present: a net is made of cells
\includetikzcentered{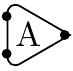}
with a fixed number of connection ports, depicted as big dots
on the picture, one of which is distinguished and called the principal port
of the cell, and of free ports, 
and of wires between those ports such that any port is
linked by exactly one wire. Then we define reduction on nets by giving rules 
of the form
$$
\includetikzinline{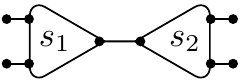}
\rightarrow
\includetikzinline{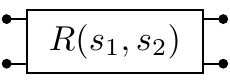}
$$
where the two cells in the left part are linked by their principal ports
and the box in the right part is a net with the same free ports as the
left part. Such a rule can be turned into a reduction of nets: as soon as a
net contains the left part we replace it with the right part.

Even though this definition is sufficient to work with interaction nets, 
it is too limited to reason on things like paths or observational equivalence.
One of the main
issues comes from the fact that we do not really know what a net is. 
The situation is quite similar for graphs: it is the author belief that
we cannot study them relying on drawings only  without being deceived
by our intuition. Thus, we are inclined to
give a precise definition of a graph as a binary relation or as a set of edges. 

The main issue to give such a definition for interaction nets is that it
should cope with reduction.
As an example consider a graph-like construction over ports and a rule
$$
\includetikzinline{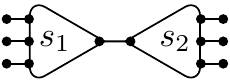}
\rightarrow
\includetikzinline{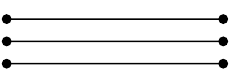}
$$
Can it be applied to the interaction net
\includetikzinline{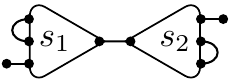}
?
If we are rigorous the left part of the rule is not exactly contained in this
net as 
\includetikzinline{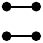}
is not contained in 
\includetikzinline{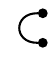}
.
Perhaps we could consider this last wire as composed of three smaller ones and 
two temporary ports like in
\includetikzinline{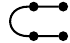}
and the whole net after reduction would be
\includetikzinline{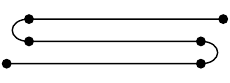}
. But then, to get back a real interaction nets we would have
to concatenate all those wires and erase the temporary ports, which would
give us the net
\includetikzinline{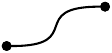}
.
We will refer to this process of wire concatenation as \emph{port fusion}.

There are many works giving definitions of interaction nets 
giving a rigorous description of reduction.
Nevertheless, they all share a common point: they deal
either implicitly or externally with port fusion. In the seminal 
article~\cite{Lafont90} a definition of nets as terms with paired variables is 
given, it is further refined in~\cite{FernandezMackie99}. In this framework 
an equivalence relation on variables deals with port fusion. 
In~\cite{Pinto00} a concrete machine is given where the computation of
the equivalence relation is broken into many steps.
A rigorous approach sharing some tools with ours is given in~\cite{Vaux07}, 
port fusion is done there by an external port rewriting algorithm.

Therefore, we raise the following question: can we give a definition of 
interaction nets allowing a simple and rigorous description of reduction
encompassing port fusion, and upon which we can prove results like 
strong confluence?
This is the aim of this paper.

Our proposition is based on
the following observation. When we plug the right part of a rule in a net,
new wires are defined based on a back and forth process between the original
net and this right part. Such kind of interaction is key to the 
\emph{geometry of interaction (GoI)}~\cite{Girard89} or 
\emph{game semantics}~\cite{AbramskyJagadeesanMalacaria94,HylandOng94}. 
The untyped nature
of interaction nets makes the former a possible way to express them. 
To be able to do so we need to express an interaction net as some kind
of partial permutation and use a composition based on the so-called
\emph{execution formula}. Such a presentation of multiplicative proof-nets 
has been made by Jean-Yves Girard in~\cite{Girard87c}.
If we try to think about the fundamental actions
one needs to be able to do on interaction nets, it is quite clear that
we can distinguish a \emph{wire action} consisting in going from one port
to another along a wire and the \emph{cell action} consisting in going 
from one cell port to another inside the same cell. Those two actions
lead to the description of a net as a pair of permutations.
One might ask whether it is possible in some case to faithfully
combine this pair in only one permutation, a solution to this question
is what one could call a GoI.

The issue of port fusion is not inherent to interaction nets and can
be found in other related frameworks. 
Diagram rewriting~\cite{Lafont03} uses a compact-closed underlying category 
allowing mathematically the \emph{straightening} of wires. 
There are strong links between this approach and ours, for example
the characterization of the free compact-closed category over a category
given in~\cite{KellyLaplaza80} shares a lot of common techniques with
our approach. It is not surprising to find such link as 
compact-closed categories are unavoidable when dealing with
\emph{geometry of interaction}. Indeed they provide -- through the Int construction 
of Joyal, Street and Verity~\cite{JoyalStreetVerity96} -- the categorical 
framework to interpret GoI~\cite{AbramskyJagadeesan92,HaghverdiScott06}.
What is different in our work, is that we
stay at a syntactical level, thus, providing a rigorous syntax
for writing and reducing programs.

This paper is organized as follows. In Section 2 we present the mathematical
tools that we are going to use. In Section 3 we define the statics of interaction
nets, in Section 4 basic tools for handling them and in Section 5 we present their
dynamics. In Section 6 we draw explicit links between interaction nets and
proof-nets. In Section 7 we present a categorical double-pushout approach to
net rewriting. In Section 8 we briefly discuss implementation of the previous 
definitions.

\section{Permutations and partial injections}
We give here the main definitions and constructions that are going to be 
central to our realization of interaction nets. Those definitions are standard
in the partial injections model of 
\emph{geometry of interaction}~\cite{Girard87c,DanosRegnier95} or in the 
definition of the traced monoidal category 
$\mathsf{PInj}$~\cite{HaghverdiScott06}.

\subsection{Permutations}
We recall that a permutation of a set $E$ is any bijection acting on $E$ and
we write $\perm{E}$ for the set of these permutations. When $E$ is finite, which we
will assume from here, for a given 
$\sigma\in\perm{E}$ we call \emph{order} the least integer $n$ such that
$\sigma^n = id_E$, for $x\in E$ we write $\Orb_\sigma(x)=\set{\sigma^i(x)}
{i\in\NN}$ and we call it the \emph{orbit} of $x$, we write $\Orbs(\sigma)$ 
for the orbits of $\sigma$. If $o$ is an orbit we write $|o|$ for its size.

We write $(c_1,\dots,c_n)$ for
the permutation sending $c_i$ to $c_{i+1}$, for $i<n$,
$c_n$ to $c_1$ and being the identity elsewhere, we call it a \emph{cycle} of 
\emph{length} $n$ which is also its order. Any permutation is a compound of 
disjoint cycles.

Let $\sigma$ be a permutation of $E$ and $\glabel$ any set, we say that 
$\sigma$ 
is labelled by $\glabel$ if we have a function $l_\sigma : \Orbs(\sigma)
\rightarrow \glabel$. We say that $\sigma$ \emph{has pointed orbits} if it is
labelled by $E$ and $\forall o\in\Orbs(\sigma)$ we have $l_\sigma(o)\in o$. 
Remark that an orbit is a sub-cycle and thus, having pointed orbits means that 
we have chosen a starting point in those sub-cycles.

\subsection{Partial injections}
A \emph{partial injection (of integers)} $f$ is a bijection from a subset
$\dom(f)$ of $\NN$, called its \emph{domain}, to a subset $\codom(f)$ of
$\NN$, called its \emph{codomain}.
We write $f : A \isoarrow B$ to say that $f$ is any partial injection such
that $\dom(f) = A$ and $\codom(f) = B$.
We write $f^\star$ for the inverse of this bijection viewed as a 
partial injection.

We call \emph{partial permutation} a partial injection $f$ such that 
$\dom(f) = \codom(f)$.

\subsection{Execution}
Let $f$ be a partial injection and $E',F'\subseteq\NN$.
We write $\frestr{f}{E'}{F'}$ for the partial injection of domain
$\set{x\in E'\cap\dom{f}}{f(x)\in F'}$ and such that
$\frestr{f}{E'}{F'}(x) = f(x)$ where it is defined.
We have
$$\frestr{f}{E'}{F'} : f^{-1}(F') \cap E' \isoarrow 
f(E') \cap F'$$
If $E'=F'$ we write $\restr{f}{E'} = \frestr{f}{E'}{E'}$.

When $\dom(f) \cap \dom(g) = \emptyset$ and $\codom(f) \cap \codom(g) = 
\emptyset$, we say that $f$ and $g$ are \emph{disjoint} and we define the
sum $f \permplus g$ and the associated refining order $\prec$ as expected.
We have $\dom(f \permplus g) = \dom(f) \uplus \dom(g)$ where $\uplus$ is
the disjoint union.

\begin{prop}
\label{fact:exec_def}
Let $f : A \uplus B \isoarrow C \uplus D$ and $g : D \isoarrow B$ a 
situation depicted by the following diagram
\includetikzinline{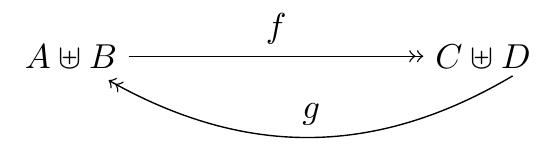}
.

\begin{enumerate}[\em i)]
    \item 
For all $n\in\NN$, the partial injection from $A$ to $C$
$$\Ex_n(f,g) = 
\frestr{f}{A}{C} \permplus \frestr{(f g f)}{A}{C} \permplus \cdots \permplus 
\frestr{(f (g f)^n)}{A}{C}$$
is well defined.
\item
$(\Ex_n(f,g))_{n\in\NN}$ is an increasing sequence of partial injections 
with respect to $\prec$, whose limit, the increasing union, is noted $\Ex(f,g)$.
\item 
If $\dom(f)$ is finite the sequence $(\Ex_n(f,g))_n$ is stationary and
$$\Ex(f,g) : A \isoarrow C$$
\end{enumerate}
\end{prop}

Fig.~\ref{fig:perm_exec} gives a graphical presentation of execution.

\begin{figure}
\centering
\includetikz{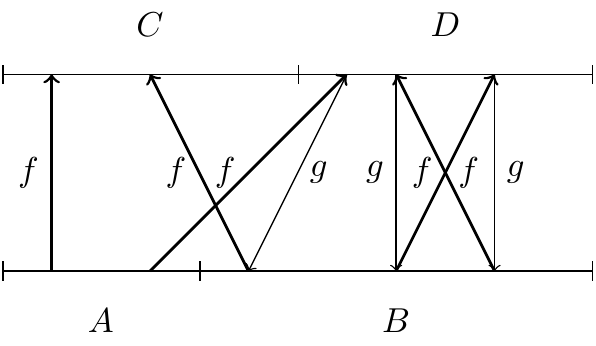}
\caption{\label{fig:perm_exec}Representation of $\Ex(f,g)$ with 
the notations of proposition~\ref{fact:exec_def}}
\end{figure}

\proof\hfill
\begin{enumerate}[i)]
\item To assert the validity of the sum all we have to have show is that
$\forall i \neq j \in \NN$ :
\begin{eqnarray*}
(f (g f)^i)(A) \cap (f (g f)^j)(A) \cap C & = & \emptyset \\
(f (g f)^i)^{-1}(C) \cap  (f (g f)^j)^{-1}(C) \cap A & = & 
\emptyset
\end{eqnarray*}

Suppose there is an
$x \in (f (g f)^i)(A) \cap (f (g f)^j)(A) \cap C$,
we set $y$ and $z \in A$ such that $x = f(g f)^i(y) = f(g f)^j(z)$.
We can further suppose that $i < j$, and we have 
$y = (g f)^{j-i}(z) \in B$, 
which is contradictory as $y \in A$ and $A \cap B = \emptyset$.

The other equality is proved in the same way.

\item Let $n \le m\in\NN$ and $x\in\dom(\Ex_n(f,g))$, by definition
of the sum there exists a unique $k$ such that 
$\Ex_n(f,g)(x) = (f(g f)^k) (x)$. But then $x\in\dom(\Ex_m(f,g))$ and
the uniqueness of $k$ asserts that $\Ex_m(f,g)(x) = (f(g f)^k)(x)$. Thus,
$\Ex_m(f,g)$ is a refinement of $\Ex_n(f,g)$.

\item Suppose there is a $x \in A - \dom(\Ex(f,g))$, then we
should have for all $k$, $(f (g f)^k)(x)\in D$ or else $\Ex(f,g)(x)$
would be defined. But $D$ being finite, there exists $n\le m$ such 
that $(f (g f)^n)(x) = (f (g f)^m)(x)$ and we get 
$x = (g f)^{m-n}(x) \in B$ which is contradictory. A simple argument
on cardinal show then that $\codom(\Ex(f,g)) = C$.\qed
\end{enumerate}

\begin{thm}[Associativity of execution]
\label{theoreme:assoc_exec} \ \\
Let 
\includetikzinline{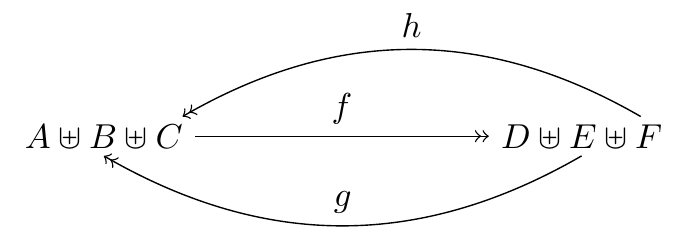}
be three partial injections.
We have $\forall n\in \mathbb{N}$
$$\Ex_n(\Ex_n(f, g), h) 
    = \Ex_n(f, g \permplus h) 
    = \Ex_n(\Ex_n(f, h), g)$$
and thus
$$\Ex(\Ex(f, g), h) 
    = \Ex(f, g \permplus h) 
    = \Ex(\Ex(f, h), g)$$
\end{thm}

\begin{proof}
Let $p \in \dom(\Ex_n(f,g\permplus h))$, there exists
$m\le n\in\mathbb{N}$ such that
\begin{eqnarray*}
\Ex_n(f,g\permplus h)(p) &=&
f ((g \permplus h) f)^m (p) \\
&=& (f (g f)^{i_1}) h \dots h (f (g f)^{i_k}) (p)
\mbox{ with } i_1 + \dots + i_k + k - 1 = m \\
&=& (\Ex_n(f,g) h \Ex_n(f,g) \dots h \Ex_n(f,g)) (p) \\
&=& (\Ex_n(f,g) (h \Ex_n(f,g))^{k-1}) (p) \\
&=& \Ex_n(\Ex_n(f,g), h)(p)
\end{eqnarray*}
By commutativity of $\permplus$ we get the other equality.
These equalities are directly transmitted to $\Ex$.
\end{proof}

This theorem is of great significance, it is a completely localized version
of Church-Rosser property. Indeed, we will see later that confluence results
are a corollary of this theorem.

The following proposition states that $\Ex$ can always be extended by an independent 
partial injection.
\begin{prop}
    \label{prop:exec sum}
    Let $f,g$ and $h$ be partial injections such that
    \begin{center}\begin{tikzpicture}
        \matrix[matrix of math nodes,column sep=3cm]{
        |(ab)| A \uplus B & |(cd)| C \uplus D \\
        };
        \draw (ab) edge[->>] node[above]{$f$} (cd);
        \draw (cd) + (0.3,-0.2) edge[bend left,->>] node[above]{$g$} (ab);
    \end{tikzpicture}\end{center}
    and $\dom(h) \cap \dom(f) = \codom(h) \cap \codom(f) = \emptyset$.

    We have $h + \Ex(f,g) = \Ex(f+h,g)$.
\end{prop}

\begin{proof}
    This result directly comes from the relation $(f+h)g(f+h) = f g f + h g f + f g h + h g h = f g f$
    as $h g = g h = 0$.
\end{proof}

\subsection{$w$-permutations and \texorpdfstring{$\Ex$}{Ex}-composition}
We call \emph{$w$-permutation} an involutive partial permutation of finite
domain. That means that a $w$-permutation is a product of disjoint cycles of
length at most 2.

Let $\sigma$ and $\tau$ be disjoint $w$-permutations 
and let $f$ be a partial injection with $\dom(f) \subseteq \dom(\sigma)$
and $\codom(f) \subseteq \dom(\tau)$.
We call the \emph{$\Ex_0$-composition} of $\sigma$ and $\tau$ along $f$ the
partial permutation $$\sigma \cexec{f} \tau = 
\Ex(\sigma\permplus\tau,f\permplus f^\star)$$

\begin{figure}
\centering
\includetikz{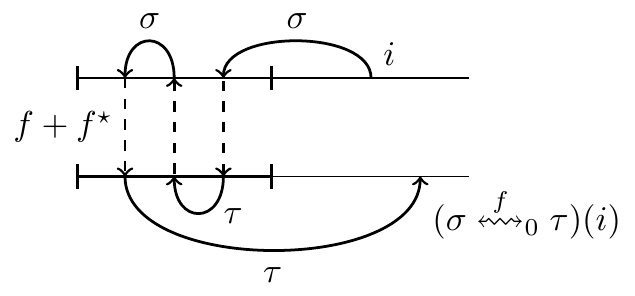}
\caption{\label{fig:perm:comp_exec}Representation of the 
$\Ex_0$-composition $\sigma \cexec{f} \tau$}
\end{figure}
Fig.~\ref{fig:perm:comp_exec} gives a representation of this composition. 

\begin{prop}
\label{prop:wperm_closed}
$\sigma\cexec{f}\tau$ is a $w$-permutation.
\end{prop}

\newcommand{\ta}{f \permplus f^\star}
\newcommand{\tb}{\sigma \permplus \tau}
\begin{proof}
    Let $x$ be an element of $\dom(\sigma \cexec{f} \tau)$, there exists $n$ such that
$$(\sigma\cexec{f}\tau)(x) = (\tb) [ (\ta) (\tb) ]^n (x)$$
Note that $(\tb)^\star = \tb$ and $(\ta)^\star = \ta$, and thus,
we have $( (\tb) [(\ta) (\tb)]^n )^\star = 
[(\tb) (\ta)]^n (\tb) = (\tb) [(\ta) (\tb)]^n$. So
$(\sigma\cexec{f}\tau)^2(x) = x$.
\end{proof}

We have $\dom(\sigma\cexec{f}\tau) = (\dom(\sigma) - \dom(f))
\uplus (\dom(\tau)-\codom(f))$. Thus, the $\Ex$ computation does not
tell anything about elements in $(\dom(\sigma)\cap\dom(f))\uplus (\dom(\tau)\cap\codom(f))$.
For such an element $x$ in 
$\dom(\sigma)\cap\dom(f)$, either there exists an $i$ in $\dom(\sigma) - \dom(f)$ with $x$ being
part of the computation of $(\sigma\cexec{f}\tau)(i)$ or there exists an $n$ such that
$(f^\star \tau f \sigma)^n(x) = x$. For those $x$ we get some kind of orbit
$O_x = \{(f^\star \tau f \sigma)^i(x)\st i\in\NN\}$ to which we get
a dual orbit for $\tau$ by setting 
$O'_x = \{(f \sigma f^\star \tau)^i( (f \sigma)(x))\st i\in\NN\}$.
By applying $f \sigma$ we get a bijective correspondence between $O_x$ and
$O'_x$.
We call \emph{double orbit} the set $O_x \cup O'_x$ and we write
$\mathbb{O}(\sigma\cexec{f}\tau)$ for the set of all double orbits.
Let $R = \{ \min_{x\in O\cup O'} x\st (O,O')\in 
\mathbb{O}(\sigma\cexec{f}\tau) \}$.
We define the full $\Ex$-composition, written $\sigma \cexece{f} \tau$, 
of domain $\dom(\sigma \cexec{f} \tau) \uplus R$ and such that
$\forall r\in R, (\sigma \cexece{f} \tau)(r) = r$.

We can now give a consequence of theorem~\ref{theoreme:assoc_exec}, stating
some kind of associativity for the $\Ex$-composition.

\begin{prop}
\label{cor:assoc_comp_exec}
Let $\sigma,\tau,\rho$ be pairwise disjoint $w$-permutations with 

\includetikz{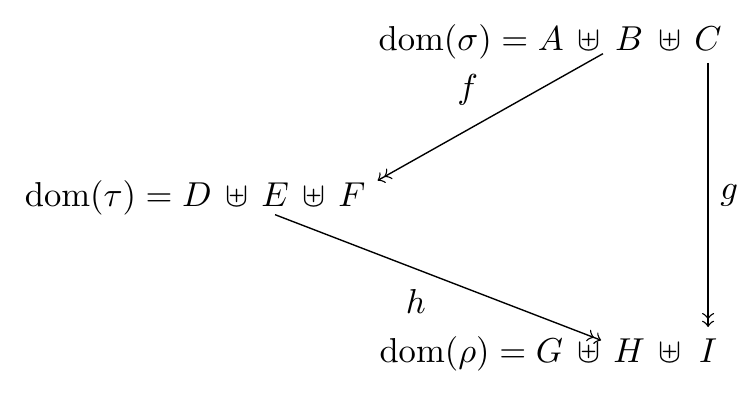}

We have
$\sigma \cexece{f \permplus g} (\tau \cexece{h} \rho)
= (\sigma \cexece{f} \tau) \cexece{g \permplus h} \rho
= (\sigma \cexece{g} \rho) \cexece{f \permplus h} \tau$.
When $h=0$ we get 
$\sigma \cexece{f \permplus g} (\tau \permplus \rho)
= (\sigma \cexece{f} \tau) \cexece{g} \rho
= (\sigma \cexece{g} \rho) \cexece{f} \tau$.
\end{prop}

\begin{proof}
    This proposition states in fact two separated results, one about
    $\Ex_0$-composition and the other about double orbits.

    We have $\sigma \cexec{f \permplus g} (\tau \cexec{h} \rho)
    = \Ex(\sigma \permplus \Ex(\tau\permplus\rho,h\permplus h^\star),
    f \permplus f^\star \permplus g \permplus g^\star)
    = \Ex(\Ex(\sigma\permplus\tau\permplus\rho,h\permplus h^\star),
    f \permplus f^\star \permplus g \permplus g^\star)$ by proposition~\ref{prop:exec sum} 
    and $= \Ex(\sigma\permplus\tau\permplus\rho,
    f \permplus f^\star \permplus g \permplus g^\star\permplus h \permplus 
    h^\star)$ by theorem~\ref{theoreme:assoc_exec}.
    We get the other equalities in the same way.

    For the equalities involving double orbits, we set
    $$\mathbb{O}_3 = \mathbb{O}(\sigma \cexec{f \permplus g} (\tau 
    \cexec{h} \rho)) - \mathbb{O}(\tau \cexec{h} \rho) 
    - \mathbb{O}(\sigma \cexec{f} \tau) - \mathbb{O}(\sigma \cexec{g} \rho)$$
    To conclude, it suffices to show that double orbits in
    $\mathbb{O}_3$ do not depend on the order of composition.
    Indeed, such an orbit is generated by an element $x$ such that
    \begin{eqnarray*}
        x & = & ((f^\star+g^\star) (\tau \cexec{h} \rho)
        (f+g) \sigma )^n(x)  \\
        & = & (\prod_{i=1}^n  (f^\star+g^\star)
        (\tau+\rho) ((h + h^\star) (\tau+\rho))^{k_i}
        (f+g) \sigma)(x) \\
        & = & (\prod_{i=1}^n F_i)(x)
    \end{eqnarray*}
    where the $F_i$s are of the following four shapes:
    $$g^\star \rho h \tau (h^\star \rho h \tau)^{k_i} f \sigma,
    f^\star \tau (h^\star \rho h \tau)^{k_i} f \sigma,
    g^\star \rho (h \tau h^\star \rho)^{k_i} g \sigma, \text{ and }
    f^\star \tau h^\star \rho (h \tau h^\star \rho)^{k_i} g \sigma$$
    When $k_i= 0$ they can have the shape 
    $g^\star \rho g \sigma$ and $f^\star \tau f \sigma$. This is enough to
    be able to group the expression by factoring $\sigma \cexec{g} \rho$ 
    or $\sigma \cexec{f} \tau$, 
    and thus, to retrieve the expressions of double orbits in any order of
    composition.
\end{proof}

\begin{rem}
    The definition of full $\Ex$-composition by means of double orbits could
    seem like a lot of trouble. We will see in the next section that the recovered
    fixpoints will allow us to interpret loops in interaction nets. One might
    argue that loops do not have to be recovered at any cost, and if our framework
    cannot see them it is for the best. In fact there are real justifications
    for loops, the main point being that seeing loops is what makes our definition 
    algebraically free. This freeness is really important as it can be seen as a 
    separation of syntax
    from semantics. A detailed discussion of the need of loops in the context
    of compact-closed categories can be found in~\cite{Abramsky05}.
\end{rem}

\section{The statics of interaction nets}
We fix a countable set $\mathcal{S}$, whose elements are called \emph{symbols},
and a function $\alpha : \mathcal{S} \rightarrow \mathbb{N}$,
the \emph{arity}. We will define nets atop $\NN$ and in this context
an integer will be called a \emph{port}.

\begin{defi}
\label{defi:reseaux_interaction}
An interaction net is an ordered pair $R = (\sigma_w,\sigma_c)$ where:
\begin{enumerate}[$\bullet$]
    \item $\sigma_w$ is a $w$-permutation.
        We write $P_l(R)$ for the fixed points of $\sigma_w$ 
        and $P(R)$ for the others, called \emph{ports of the net} $R$.
    \item $\sigma_c$ is a partial permutation of $P(R)$ with
        pointed orbits and labelled by $\mathcal{S}$ in such a way
        that $\forall o\in\Orbs(\sigma_c), |o|=\alpha(l(o)) + 1$ where 
        $l$ is the labelling function.
\end{enumerate}
\end{defi}

The elements of $P_l(R)$ are called \emph{loops} and the other orbits of 
$\sigma_w$, which are necessarily of length $2$, are called \emph{wires}. 
The domain of $\sigma_w$ is called the \emph{carrier} of the net.
We write $P_c(R) = \dom(\sigma_c)$, whose elements are called \emph{cell ports}, and $P_f(R) = P(R) - P_c(R)$, whose elements are called \emph{free ports}.

An orbit of $\sigma_c$ is called a \emph{cell}. We write $\pal$ for the 
pointing function of $\sigma_w$.
Let $c$ be a cell, $\pal(c)$ is its \emph{principal port}
and for $i<|c|$ the element
$(\sigma_c^i \circ \pal)(c)$ is its \emph{$i$th auxiliary port}. 

Note that a port of a net is present in exactly one wire and at most one cell.

\subsection{Representation}
Nets admit a very natural representation. We shall draw a cell of symbol
$A$ as a triangle 
\includetikzinline{intro_cell}
where the principal port is the dot on the apex and
auxiliary ports are lined up on the opposing edge. We draw free ports
as points. To finish the drawing we add a line between any two ports
connected by a wire, and draw circles for loops.

As an example consider the net $R = (\sigma_w,\sigma_c)$ with
$$\sigma_w = (1)(2\ 3)(4\ 5)(6\ 7)(8\ 9) \text{ and } 
\sigma_c=(\point{4}\ 3)_A(\point{5}\ 6\ 7)_B$$
where permutations are given by cycle decomposition and
$(\point{c_1}\ c_2\ \dots\ c_n)_S$ is a cell of point $c_1$ and
symbol $S$.
This net will have the representation 
\includetikzcentered{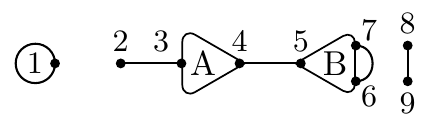}

\subsection{Morphisms of nets and renaming}
\begin{defi}
Let $R = (\sigma_w,\sigma_c)$ and $R' = (\sigma'_w,\sigma'_c)$ be two
interaction nets. A function $f : P(R) \mapsto P(R')$ is a 
\emph{morphism from $R$ to $R'$} if and only if 
$$f \circ \sigma_w = \sigma_w' \circ f,\quad
f(P_c(R)) \subseteq P_c(R'),$$
$$\forall p \in P_c(R), (f \circ \sigma_c)(p)
= (\sigma'_c \circ f)(p),$$ and $\forall o\in\Orbs(\sigma_c)$
we have
$(f\circ\pal)(o) = (\pal\circ f)(o)$ and $l(o) = (l\circ f)(o)$.
When $f$ is the identity on $P_f(R)$ it is said to be an
\emph{internal morphism}.
\end{defi}

\begin{exa}
Consider the net:
$$R = ( (1\ 2)(3\ 4)(5\ 6)(7\ 8)(9\ 10)(11\ 12)(13\ 14),
   (\point{2}\ 3\ 5)_A (\point{8}\ 9\ 11)_A )$$
of representation:
\begin{ctp}
    \matrix[column sep=1cm]{
    \inetcell(A){A} & \inetcell(Ap){A} & \node (p14) {}; \\
    };
    \scriptsize
    \inetwirefree(A.pal)(A.left pax)(A.right pax)
        (Ap.pal)(Ap.left pax)(Ap.right pax);
    \inetwiredraw (p14) node[above]{$14$}
    -- +(0,-.5) node(p15){} node[below]{$13$};

    \inetport(p14)(p15)
    \node[left] at (A.right pax) {$3$};
    \node[right] at (A.left pax) {$5$};
    \node[left] at (A.pal) {$2$};
    \node[below] at (A.above pal) {$1$};
    \node[above] at (A.above right pax) {$4$};
    \node[above] at (A.above left pax) {$6$};

    \node[left] at (Ap.right pax) {$9$};
    \node[right] at (Ap.left pax) {$11$};
    \node[left] at (Ap.pal) {$8$};
    \node[below] at (Ap.above pal) {$7$};
    \node[above] at (Ap.above right pax) {$10$};
    \node[above] at (Ap.above left pax) {$12$};
\end{ctp}
and the net:
$$S = ( (1\ 2)(3\ 4)(5\ 6)(7\ 8),
    (\point{2}\ 3\ 5)_A (\point{8})_B )$$
of representation:
\begin{ctp}
    \matrix[column sep=1cm]{
    \inetcell(A){A} & \inetcell[arity=0](B){B}  \\
    };
    \scriptsize
    \inetwirefree(A.pal)(A.left pax)(A.right pax)
        (B.pal);

    \node[left] at (A.right pax) {$3$};
    \node[right] at (A.left pax) {$5$};
    \node[left] at (A.pal) {$2$};
    \node[below] at (A.above pal) {$1$};
    \node[above] at (A.above right pax) {$4$};
    \node[above] at (A.above left pax) {$6$};

    \node[left] at (B.pal) {$8$};
    \node[below] at (B.above pal) {$7$};
\end{ctp}

Let $f$ be the application defined by $f(1) = f(7) = f(13) = 1$, 
$f(2) = f(8) = f(14) = 2$, $f(3) = f(9) = 3$, $f(5) = f(11) = 5$, 
$f(4) = f(10) = 4$ and $f(6) = f(12) = 6$. This is a morphism from $R$
to $S$.
\end{exa}

\begin{rem}
The equality $f \circ \sigma_w = \sigma_w' \circ f$
seems quite strong, but could in fact be deduced from a simple inclusion
of functional graphs, $f \circ \sigma_w
\subseteq \sigma_w' \circ f$. Indeed, let $(p,p')$ be in the graph of
$\sigma_w' \circ f$, we can compute $(f \circ \sigma_w)(p)$
which by the inclusion cannot be anything else than $p'$.
\end{rem}

Let us detail a bit more this definition. We note that for any two
partial permutations $\sigma$ and $\tau$, the equation $f \circ \sigma = 
\sigma' \circ f$ induces that a $o\in \Orbs(\sigma)$ is mapped to 
an element $f(o)\in\Orbs(\sigma')$ such that $|f(o)|$ is a divisor of $|o|$.

In this case a loop is sent to a loop, a wire
to a loop or a wire, and a cell to another cell. The last two equations
say that the principal port of a cell is mapped to a principal port, and
symbols are preserved. So a cell is mapped to a cell of same arity, and 
each port is mapped to the same type of port. Moreover only a wire
linking free ports can be mapped to a loop or any kind of wire. As 
soon as the wire is linking one cell port the third condition on the
morphism must send it to a wire of the same type.

With those facts, it is natural to call \emph{renaming} 
(resp. \emph{internal renaming}) an isomorphism (resp. internal isomorphism).
An isomorphism class captures interaction nets as they are drawn on paper.
On the other hand, an internal isomorphism class corresponds to
interaction nets drawn where we have also given distinct names to free
ports, hence the name internal.
This is an important notion because the drawing
\includetikzinline{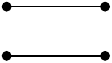}
is the same as
\includetikzinline{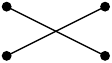}
Whereas the drawing
\includetikzinline{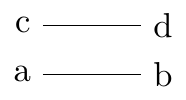}
is different from
\includetikzinline{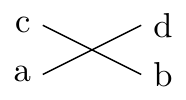}.

In fact, as soon as we would like to consider nets as some kind of terms, we
will have to consider them up to internal isomorphism. Free ports
correspond to free variables, whereas cell ports correspond to bound
variables. For example the $\lambda$-term
$\lambda x. (x) y$ is of course the same as $\lambda z. (z) y$ but it is distinct
from $\lambda x. (x) z$.

\begin{rem}
\label{rmq:disjoint carriers}
Given the fact that nets have finite carriers we can always consider
that two nets have disjoint carriers up to renaming.
\end{rem}

\section{Tools of the trade}
We give here the main tools that are going to be crucial to our definition
of reduction.
\subsection{Gluing and cutting}
\begin{defi}
Let $R = (\sigma_w, \sigma_c)$ and $R' = (\sigma_w', \sigma_c')$ be two
nets with disjoint carriers\footnote{Which is not a loss of generality 
thanks to remark~\ref{rmq:disjoint carriers}.} and let
$f$ be a partial injection of domain included in $P_f(R)$ and codomain
included in $P_f(R')$.
We call \emph{gluing of $R$ and $R'$ along $f$} the net
$R\glue{f}R' = (\sigma_w \cexece{f} \sigma_w', \sigma_c \permplus \sigma'_c)$.
\end{defi}

From this definitions we get the following obvious facts:
\begin{eqnarray*}
P(R \glue{f} R') &=& (P(R) - \dom(f)) \uplus (P(R') - \codom(f)) \\
P_c(R \glue{f} R') &=& P_c(R) \uplus P_c(R') \\
P_f(R \glue{f} R') &=& 
    (P_f(R) - \dom(f)) \uplus (P_f(R') - \codom(f))\\
    R \glue{f} R' &=& R' \glue{f^\star} R
\end{eqnarray*}
For the special case of gluing where $f=0$ we have 
$R\glue{0}R'=(\sigma_w\permplus\sigma_w',\sigma_c\permplus\sigma'_c)$, we write
this special kind of gluing $R\dglue R'$, it is the so-called parallel 
composition of the two nets.

Fig.~\ref{fig:comp_reseaux} gives a representation of gluing.

\begin{prop}
\label{fact:gluing identity}
If $R = R \glue{f} R'$ then $f=0$ and $R'=\mathbf{0} = (0,0)$. If $\mathbf{0}
= R \glue{f} R'$ then $f = 0$ and $R = R' = \mathbf{0}$.
\end{prop}

\begin{proof}
    We will only prove the first assertion, the second being similar.
    It is a direct consequence of the previous facts, $R'$ must have
    no cells, no free ports and no loops. The only net having this property
    is the empty net $\mathbf{0}$.
\end{proof}

We can get some kind of associativity property for gluing.
\begin{prop}
\label{prop:assoc_comp_reseau}
Let $R=(\sigma_w,\sigma_c)$,
$S=(\tau_w,\tau_c)$ and $T=(\rho_w,\rho_c)$ 
be nets of disjoint carriers and let $f,g$ and $h$ be partial injections
satisfying the diagram of proposition~\ref{cor:assoc_comp_exec} with respect to 
$\sigma_w, \tau_w$ and $\rho_w$.

We have
$R \glue{f \permplus g} (S \glue{h} T)
= (R \glue{f} S) \glue{g \permplus h} T
= (R \glue{g} T) \glue{f \permplus h} S$.
\end{prop}

\begin{proof}
The wire part of the equality is a restriction of 
proposition~\ref{cor:assoc_comp_exec} and the cell part is 
the associativity of $\permplus$.
\end{proof}
The following corollary will often be sufficient.
\begin{cor}
    \label{cor:assocgluing}
    If we have a decomposition $R_0 = R \glue{f} (S \glue{g} T)$ then
    there exists $f_S, f_T$ such that $R_0 = (R \glue{f_S} S) 
    \glue{g+f_T} T$.
\end{cor}

\begin{figure}
\centering
\includetikzinline{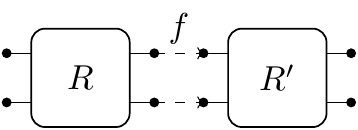}
\quad
$\longrightarrow$
\quad
\includetikzinline{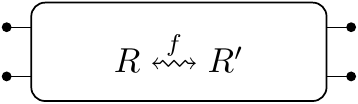}
\caption{\label{fig:comp_reseaux}Representation of 
the gluing of two interaction nets}
\end{figure}
\begin{figure}
\centering
\subfloat[]{
\includetikz{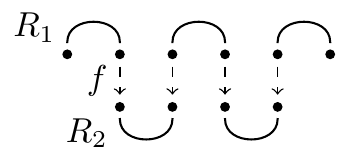}
} \quad
\subfloat[]{
\includetikz{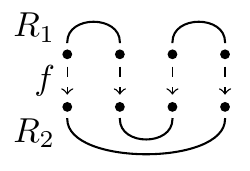}
}
\caption{\label{fig:cutting}
Representation of two special cuttings: (a) a cutting of a single wire and
(b) a cutting of a loop}
\end{figure}
We can use the gluing to define dually the notion of cutting a subnet of
an interaction net.

\begin{defi}
Let $R$ be a net, we call \emph{cutting of $R$} a triple $(R_1, f, R_2)$
such that $R = R_1 \glue{f} R_2$. Any net $R'$ 
appearing in a cutting of $R$ is called a \emph{subnet of $R$}, noted
$R' \subseteq R$.
\end{defi}

Fig.~\ref{fig:cutting} gives an example of cutting. The fact that
we can cut many times a wire or that we can divide a loop in many
wires hints at the complexity behind these definitions.

\begin{prop}
The relation $\subseteq$ is an ordering of nets.
\end{prop}
\begin{proof} The relation $\subseteq$ is \textbf{reflexive}: 
    $R = R \glue{0} \mathbf{0}$ and thus, $R \subseteq R$.

It is \textbf{antisymmetric}: let $R_1$ and $R_2$ be nets such that
$R_1\subseteq R_2$ and $R_2 \subseteq R_1$. We have
$R_1 = R_2 \glue{f} R'_2$ and $R_2 = R_1 \glue{g} R'_1$.

So $R_2 = (R_2 \glue{f} R'_2) \glue{g} R'_1$.
By applying the corollary~\ref{cor:assocgluing} we get
$R_2 = R_2 \glue{f_1} (R'_2 \glue{g+f_2} R'_1)$.
and by applying the proposition~\ref{fact:gluing identity} twice 
we get $R'_2 = R'_1 = \mathbf{0}$. So $R_1 = R_2$.

And it is
\textbf{transitive}: let $R\subseteq S \subseteq T$, then 
$S = R \glue{f} R'$ and $T = S \glue{g} S'$, so
$T = (R \glue{f} R') \glue{g} S'$. By applying the 
corollary~\ref{cor:assocgluing} we have $T = R \glue{f_1} (R' 
\glue{g+f_2} S')$, that is to say $R \subseteq T$.
\end{proof}

\subsection{Extending morphisms by gluing}
\begin{prop}
If $\alpha : R \rightarrow S$ is a morphism of nets, and $T =
S \glue{f} S'$, then there exists a morphism $\widehat{\alpha} : R \rightarrow T$
extending $\alpha$.
\end{prop}
\begin{proof}
It is obvious how to define the image of a cell in $R$ into $T$, because
$\alpha$ maps it to a cell in $S$ and cells are preserved by gluing.
So, the only thing to prove is that we can properly define the image of
a wire in $R$. We consider a wire $(p\ p')$ in $R$ which is mapped to
another wire $(\alpha(p)\ \alpha(p'))$ in $S$ (the case where it is a loop is
trivial as loops are also preserved by gluing). In $T$ this wire has either become
a loop, and thus we send, by $\widehat{\alpha}$, $p$ and $p'$ to the loop port,
or it has become a wire trough the $\Ex$-composition:
$$q \rightarrow ... \rightarrow \alpha(p) \xrightarrow{\sigma_w} \alpha(p') 
\rightarrow ... \rightarrow q'$$
where $S = (\sigma_w,\sigma_c)$,
in which case we define $\widehat{\alpha}(p) = q$ and $\widehat{\alpha}(p') = q'$.

By construction, $\widehat{\alpha}$ is a morphism.
\end{proof}

With $id_R : R \rightarrow R$ being the identity function on ports, and
$T = R \glue{f} S$, we simply write $R \subseteq T$ for the morphism $\widehat{id_R}$
which we refer to as the inclusion map of $R$ into $T$. Note that in our setting
these maps are not just co-extensions of identity, this is due to 
our notion of subnets.

\begin{defi}
We say that $\alpha : R \rightarrow S$ is \emph{almost injective} when there exists
a decomposition $S = \beta(R) \glue{f} R'$ with $\beta$ a renaming and $\widehat{\beta} = 
\alpha$ where $\widehat{\beta}$ is given by the previous proposition. We also
use the notation $\widetilde{\alpha} = \beta$.
\end{defi}

Inclusion maps are the archetypal almost injective morphisms. Indeed, every 
almost injective morphism splits as a renaming followed by an inclusion map.

\subsection{Interfaces and contexts}
To define reduction by using the subnet relation, it would be easier if
we could refer implicitly to the identification function in a gluing.
As an intuition, consider terms contexts with multiple holes, to substitute
completely such contexts we could give a function from holes to terms and
fill them accordingly. But a more natural definition would be to give a
distinct number to each hole and to fill based on a list of terms. 
The substitution would give the 
first term to the first hole, and so on. The following definition is
a direct transposition of this idea in the framework of interaction nets.

\begin{defi}
    We call \emph{interface of a net $R$} a subset $I=\{p_1,\dots,p_n\}$ of 
    $P_f(R)$
together with a linear ordering, the length of the order chain 
$p_1 < \dots < p_n$ is called
the \emph{size}. We say that $R$ contains the interface $I$, noted 
$I \icontains R$. An interface is \emph{canonical} if it contains all
the free ports of a net.

Let $I$ and $I'$ be disjoint interfaces of the same net, we write 
$I I'$ the union of these subsets ordered by the concatenation of the two 
order chains. Precisely 
$x \le_{II'} y \iff x \le_{I} y$ or $x \le_{I'} y$ or $x \in I \wedge
y \in I'$.

Let $I$ and $I'$ be two interfaces of same size, there exists one and only
order-preserving bijection from $I$ to $I'$ that we write $\cord{I}{I'}$
and call the \emph{chord between $I$ and $I'$}.

We call \emph{context} a pair $(R,I)$ where $I$ is an interface 
contained in the net $R$, it is written $\context{R}{I}$.

Let $\context{R}{I}$ and $\context{R'}{I'}$ be two contexts with 
interfaces of same size, we write $$\context{R}{I} \cglue \context{R'}{I'}
= R \glue{\cord{I}{I'}} R'$$
\end{defi}

In the following when we write $\context{R}{I} \cglue \context{R'}{I'}$ we
implicitly assume that $I$ and $I'$ are of same size.

We now can state commutativity of gluing directly, the proof being trivial.
\begin{prop}
$\context{R}{I} \cglue \context{R'}{I'} =
\context{R'}{I'} \cglue \context{R}{I}$\qed
\end{prop}

The following trivial fact asserts that any gluing can be seen as a context 
gluing.
\begin{prop}
Let $R \glue{f} R'$ be a gluing, there exist interfaces $I\icontains R$
and $I'\icontains R'$ such that
$R \glue{f} R' = \context{R}{I} \cglue \context{R'}{I'}$.
\end{prop}

\begin{proof}
It suffices to take $I = \dom(f)$ with any linear ordering, and to define
the only ordering of $I' = \codom(f)$ such that $f$ is strictly increasing.
\end{proof}

\begin{cor}
$R_1 \subseteq R \iff \exists I_1, R_2, I_2$ such that 
$R = \context{R_1}{I_1} \cglue \context{R_2}{I_2}$.
\end{cor}

We can now restate corollary~\ref{cor:assocgluing} with interfaces:
\begin{cor}
\label{cor:associativity_gluing}
For all nets $R,S,T$ and interfaces $I,J,K,L$,
there exists interfaces $I'$, $J'$, $K'$, $L'$ such that
$$\context{R}{I} \cglue 
\context{(\context{S}{J} \cglue \context{T}{K})}{L}
= \context{(\context{R}{I'} \cglue \context{S}{J'})}{L'}
\cglue \context{T}{K'} .
$$
\end{cor}

\section{Dynamics}
\label{sec:dynamics}
Given the previous definitions we will now present the dynamics of net. It should be remarked
that our definition of dynamics is quite similar to the usual one: it amounts to finding
a subnet called a redex and substituting it with another subnet. The main difference lies in our 
rigorous definition of subnets.

\begin{defi}
    Let $s_1$ and $s_2$ be symbols.
    We call \emph{interaction rule for $(s_1,s_2)$} a couple 
    $(\context{R_r}{I_r},\context{R_p}{I_p})$ where 
    $$
        R_r =
        \left(
        \begin{array}{l} 
            (b\ c)(a_1\ b_1)\dots(a_n\ b_n)(c_1\ d_1)\dots(c_m\ d_m), \\ 
            (\point{b}\ b_1\ \dots\ b_n)_{s_1} (\point{c}\ c_1\ \dots\ c_m)_{s_2} 
    \end{array} \right)
    $$
    and $I_r$ and $I_p$ are both canonical -- comprised of all free ports --
    and of same size.

    Let $\mathcal{R} = (\context{R_r}{I_r}, \context{R_p}{I_p})$ be a rule.
    we call \emph{reduction by $\mathcal{R}$}
    the binary relation 
    $\xrightarrow{\mathcal{R}}$
    on nets such that 
    for all renaming $\alpha$ and $\beta$, and for all net $S$ with
    $S = \context{R}{I} \cglue \context{\alpha(R_r)}{\alpha(I_r)}$
    we set $S \xrightarrow{\mathcal{R}} S'$ where $S' = \context{R}{I}
    \cglue \context{\beta(R_p)}{\beta(I_p)}$.
\end{defi}

The net $R_r$ has the representation \includetikzinline{redex_repr}.
Remark that the reduction is defined as soon as a net contains a
renaming of the redex $R_r$. This reduction appears to be non-deterministic
but it is only the expansion of a deterministic reduction to cope with 
all possible renamings.

We recall now the formal definition of the main property of interaction nets and
we wish that our definition ensures it.

\begin{defi}
    Let $\xrightarrow{\mathcal{R}}$ be a binary relation on a set $E$, we say that
    it is \emph{strongly confluent} if and only if for all
    $x, y, z \in E$ such that $y \neq z$ and
    $y \xleftarrow{\mathcal{R}} x \xrightarrow{\mathcal{R}} z$  and
    there exists $t \in E$ with
    $y \xrightarrow{\mathcal{R}} t \xleftarrow{\mathcal{R}} z$.
\end{defi}

\begin{prop}
    \label{prop:strong_confluent}
    Let $R$ be a net and $\mathcal{R}_1$, $\mathcal{R}_2$ be two
    interaction rules applicable on $R$ on distinct redexes such that
    $R_1 \xleftarrow{\mathcal{R}_1} R \xrightarrow{\mathcal{R}_2} R_2$
    and all the ports both in $R_1$ and $R_2$ are also in $R$.
    There exists a net $R'$ such that
    $R_1 \xrightarrow{\mathcal{R}_2} R' \xleftarrow{\mathcal{R}_1} R_2$.
\end{prop}

\begin{proof}
    For $i=1,2$, set $\mathcal{R}_i = 
    (\context{R_{r,i}}{I_{r,i}}, \context{R_{p,i}}{I_{p,i}})$.
    The shape of redexes allow us to assert that if they are distinct
    then they are disjoint. As $R$ contains both a redex $\alpha_1(R_{r,1})$ 
    and a redex $\alpha_2(R_{r,2})$, then
    we can deduce that $\alpha_1(R_{r,1}) \dglue \alpha_2(R_{r,2}) \subseteq 
    R$. More precisely we have 
    $$R = \ccglue{(\alpha_1(R_{r,1}) \dglue 
    \alpha_2(R_{r,2}))}{\alpha_1(I_{r,1}) \alpha_2(I_{r,2})}{R_0}{I}$$
    We get 
    $$R_1 = \ccglue{(\beta_1(R_{p,1}) \dglue 
    \alpha_2(R_{r,2}))}{\beta_1(I_{p,1}) \alpha_2(I_{r,2})}{R_0}{I}$$
    for a renaming $\beta_1$, and the same kind of expression for $R_2$. 
    It is straightforward to check that the net 
    $$R' = \ccglue{(\beta_1(R_{p,1}) \dglue 
    \beta_2(R_{p,2}))}{\beta_1(I_{p,1}) \beta_2(I_{p,2})}{R_0}{I}$$
    satisfies the conclusion by applying 
    proposition~\ref{prop:assoc_comp_reseau}. 
    The very existence of this net relies on 
    the disjointness of the $\beta_i(R_{p,i})$ which is ensured by the
    hypothesis on ports contained in both $R_1$ and $R_2$.
\end{proof}

\begin{cor}
    Let $\mathcal{L}$ be a set of rules such that for any pair of
    symbols there is at most one rule over them. The reduction
    $\xrightarrow{\mathcal{L}} = \bigcup_{\mathcal{R}\in\mathcal{L}}
    \xrightarrow{\mathcal{R}}$ is strongly confluent up to a renaming.
\end{cor}

By \emph{up to a renaming} we mean that we might have to rename
one of the nets in a critical pair before joining them. This is due to the
disjointness condition in proposition~\ref{prop:strong_confluent}. Remark
that we can always substitute one of the branch of the critical pair by
another instance of the same rule on the same redex in such a way
that this condition is ensured.

\subsection{Example}
\label{sec:example_mll}
We will now give a thorough example of a net reduction using the Multiplicative Linear Logic symbols and rules.
We display representations next to the net definitions.\footnote{Nevertheless, these representations are not required to 
do the reduction, they are merely here to help the reader.}

Let us consider the rule $\mathcal{R} = (\context{R_r}{I_r},\context{R_p}{I_p})$ where 
$$R_r = \left(\begin{array}{l} 
        (0\ 3)(6\ 1)(7\ 2)(8\ 4)(9\ 5), \\
        (\point{0}\ 1\ 2)_{\Par} (\point{3}\ 4\ 5)_{\Times} 
\end{array} \right) 
\quad
\includetikzinline{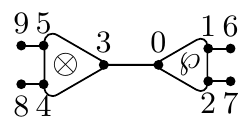}
\quad I_r = 6 < 7 < 8 < 9$$

$$R_p = \left(\begin{array}{l} 
        (10\ 12)(11\ 13), \\
        0
\end{array} \right) 
\quad
\includetikzinline{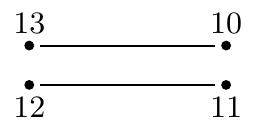}
\quad I_p = 10 < 11 < 12 < 13$$
Now let $R$ be the net
$$\left(\begin{array}{l} 
        (0\ 3)(1\ 2)(8\ 4)(7\ 5)(6\ 9), \\
        (\point{0}\ 1\ 2)_{\Par} (\point{3}\ 4\ 5)_{\Times} (\point{6}\ 7\ 8)_{\Par}
\end{array} \right)
\quad
\includetikzinline{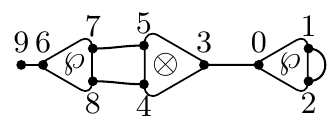}
$$
It can be expressed as 
$$
\left(\!\begin{array}{l} 
        (6\ 9)(7\ 10)(8\ 11)(12\ 13), \\
        (\point{6}\ 7\ 8)_{\Times}
\end{array}\!\right)^{\!10 < 11 < 12 < 13} \cglue
\left(\!\begin{array}{l} 
        (0\ 3)(16\ 1)(17\ 2)(15\ 4)(14\ 5), \\
        (\point{0}\ 1\ 2)_{\Par} (\point{3}\ 4\ 5)_{\Times} 
\end{array}\!\right)^{\!14 < 15 < 16 < 17}
$$
$$
\includetikzinline{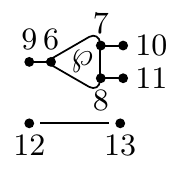}
\quad \cglue \quad
\includetikzinline{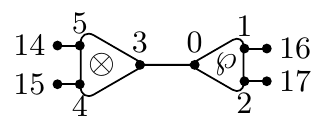}
$$
the latter context being a renaming of $\context{R_r}{I_r}$, which we substitute with the following
renaming of $\context{R_p}{I_p}$:
$$\left(\begin{array}{l} 
        (14\ 16)(15\ 17), \\
        0
\end{array} \right)^{14 < 15 < 16 < 17}
\quad
\includetikzinline{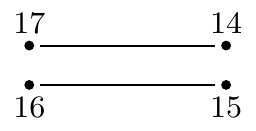}
$$
Thus, we get the net
$$
\left(\begin{array}{l} 
        (6\ 9)(7\ 10)(8\ 11)(12\ 13), \\
        (\point{6}\ 7\ 8)_{\Times}
\end{array} \right)^{10 < 11 < 12 < 13} \cglue
\left(\begin{array}{l} 
        (14\ 16)(15\ 17), \\
        0
\end{array} \right)^{14 < 15 < 16 < 17}$$
which simplifies into
$$
\left(\begin{array}{l} 
        (6\ 9)(7\ 8), \\
        (\point{6}\ 7\ 8)_{\Times}
\end{array}\right)
\quad
\includetikzinline{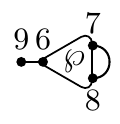}
$$

\section{Interaction nets are the \texorpdfstring{$\Ex$}{Ex}-collapse of Axiom/Cut nets}
We introduce now a notion of nets lying between proof-nets of multiplicative
linear logic and interaction nets. When we plug directly two interaction nets
a complex process of wire simplification occurs. When we plug two proof-nets
we only add special wires called \emph{cuts} and we have an external
notion of reduction performing such simplification. In this section we define 
nets with two kinds of wires: \emph{axioms} and \emph{cuts}. Those nets
allow us to give a precise account of the folklore assertion that interaction
nets are a quotient of multiplicative proof-nets.
\subsection{Definition and juxtaposition}
\begin{defi}
    An \emph{Axiom/Cut net}, \emph{AC net} for short,
    is a tuple $R = (\sigma_A, \sigma_C,\sigma_c)$ where:
    \begin{enumerate}[$\bullet$]
        \item $\sigma_A$ and $\sigma_C$ 
            are $w$-permutations of 
            finite domain such that
             $\dom(\sigma_C) \subseteq \dom(\sigma_A)$, 
             $\sigma_C$ has no fixed points and 
                    if $(a\ b)$ is an orbit of $\sigma_C$ then there exists
                    $c\neq a$ and $d\neq b$ such that $(c\ a)$ and $(b\ d)$ are
                    orbits of $\sigma_A$.

            We write $P_l(R)$ for the fixed points of $\sigma_A$ and 
            $P(R) = \dom(\sigma_A) - \dom(\sigma_C) - P_l(R)$.
        \item $\sigma_c$ is an element of 
            $\mathfrak{S}(P_c(R))$, where $P_c(R) \subseteq P(R)$,
            has pointed orbits and is labelled by $\mathcal{S}$ in such a way
            that $\forall o\in\Orbs(\sigma_c), |o|=\alpha(l(o))$ where 
            $l$ is the labelling function.
    \end{enumerate}
\end{defi}
\noindent The orbits of $\sigma_C$, called \emph{cuts},
are some kind of undirected unary cells linking orbits of $\sigma_A$, 
called \emph{axioms}.

We directly adapt the representation of interaction nets to AC nets by
displaying $\sigma_c$ as double edges. For example the AC net $R = 
(\sigma_A, \sigma_C, \sigma_c)$ with
$$\sigma_A = (1\ 2)(3\ 4)(5\ 6), \sigma_C = (2\ 3), \sigma_c = (\point{4}\ 5)_S$$
will be represented by
\begin{ctp}
    \inetcell(c){$S$}[R]
    \coordinate (cap) at ($(c.pal)!2!(c.above pal)$);
    \coordinate (cp) at ($(c.pal)!4!(c.above pal)$);
    \coordinate (cpp) at ($(c.pal)!6!(c.above pal)$);
    \coordinate (cam) at ($(c.middle pax)!2!(c.above middle pax)$);
    \draw (cap) edge[bend right,double,thick] (cp);
    \inetwiredraw (cp) -- (cpp) (c.pal) -- (cap) (cam) -- (c.middle pax);

    \inetport(c.pal)(c.middle pax)(cp)(cpp)(cam)(cap)

\end{ctp}

We can adapt most of the previous definitions for those nets, most importantly
free ports, interfaces and contexts.
The nice thing about AC nets is that they yield a very simple composition.
\begin{defi}
    Let $\context{R}{I} = (\sigma_A, \sigma_C, \sigma_c)$ and 
    $\context{R'}{I'} = (\tau_A, \tau_C, \tau_c)$ be two contexts on AC nets
    with disjoint carriers, with $I = i_1 > \dots > i_n$ and 
    $I' = i'_1 > \dots > i'_n$.

    We call \emph{juxtaposition of $\context{R}{I}$ and $\context{R'}{I'}$}
    the AC net 
    $$\context{R}{I} \ACglue \context{R'}{I'} = (\sigma_A \permplus \tau_A,
    \sigma_C \permplus \tau_C \permplus 
    (i_1\ i'_1) \dots (i_n\ i'_n), \sigma_c \permplus \tau_c)$$
\end{defi}

The juxtaposition is from the logical point of view a generalized cut, and its
interpretation in terms of permutation is exactly the definition made by
Girard in~\cite{Girard87c}.
    
\subsection{\texorpdfstring{$\Ex$}{Ex}-collapse}
\begin{prop}
    Let $R = (\sigma_A, \sigma_C, \sigma_c)$ be an AC net
    and $f$ be a partial injection such that
    $\dom(\sigma_C) = \dom(f)$ and $\codom(f) \cap \dom(\sigma_A) = \emptyset$.

    The couple $(\sigma_A \cexece{f} f \circ \sigma_C \circ f^\star, \sigma_c)$,
    is an interaction net.

    It does not depend on $f$ and we call it the $\Ex$-collapse of $R$, 
    noted $\exc{R}$.
\end{prop}

\noindent  For the definition of the $\Ex$-composition to be correct, we have to
delocalize $\sigma_C$ to a domain disjoint from $\dom(\sigma_A)$. 
The $\Ex$-collapse amounts to replace any maximal chain 
$a_1 \xrightarrow{\sigma_A} b_1 \xrightarrow{\sigma_C} a_2 \dots
b_{n-1} \xrightarrow{\sigma_A} a_n$
by a chain 
$a_1 \xrightarrow{\sigma_A} b_1 \xrightarrow{f} f(b_1) 
\xrightarrow{f \circ \sigma_C \circ f^\star} f(a_2)
\xrightarrow{f^\star} \dots
b_{n-1} \xrightarrow{\sigma_A} a_n$
and then to compute the $\Ex$-composition to get
$a_1 \xrightarrow{\sigma_A \cexece{f} \sigma_C} a_n$.

\begin{proof}
    Remark that for this to be an interaction net, the only property to
    be checked which is not a direct consequence of the definition of AC
    nets is the fact that
    $\sigma_A \cexece{f} f \circ \sigma_C \circ f^\star$ is a $w$-permutation,
    but this comes directly from proposition~\ref{prop:wperm_closed}.

    This remark asserts that $f$ as only a shallow role in the definition.
    Indeed, every time $f$ is applied in the $\Ex$-composition, it is
    followed by an application of its inverse.
    Moreover, for partial injections
    $\sigma_1, \tau_1, \dots, \sigma_n, \tau_n$, we have
    $$\sigma_n \circ \tau_n \circ \dots \circ \sigma_1 \circ \tau_1 = 
    \sigma_n \circ f_n^\star \circ f_n \circ \tau_n \circ g_n^\star
    \circ g_n \circ \dots \circ g_1^\star \circ g_1 \circ \sigma_1
    \circ f_1^\star \circ f_1 \circ \tau_1$$
    for every partial injections $f_1,g_1,\dots,f_n,g_n$ such that
    $\dom(f_i) \subseteq \codom(\tau_i) \cap \dom(\sigma_i)$
    and $\dom(g_i) \subseteq \codom(\sigma_i) \cap \dom(\tau_i)$
\end{proof}

\begin{prop}
    For each interaction net $R$ there exists a unique AC net $R'$ of the form
    $(\sigma_A, 0, \sigma_c)$ such that $\exc{R'} = R$. $R'$ is said to be
    \emph{cutfree}.
\end{prop}

\begin{proof}
    If $R = (\tau_w,\tau_c)$ we only have to take $R'= (\tau_w, 0, \tau_c)$.
    Uniqueness comes from the fact that $\sigma \cexece{0} 0
    = \sigma$.
\end{proof}

\begin{defi}
    Let $R$ and $R'$ be two AC nets, we say that 
    \emph{$R$ and $R'$ are $\Ex$-equivalent}, noted $R \exequiv R'$ when
    $\exc{R} = \exc{R'}$.
\end{defi}

We have an obvious correspondence between juxtaposition and gluing.
\begin{prop}
    $\exc{\context{R}{I} \ACglue \context{R'}{I'}}
    = \context{\exc{R}}{I} \cglue \context{\exc{R'}}{I'}$
\end{prop}

\begin{proof}
   We set $R = (\sigma_A, \sigma_C, \sigma_c)$, $R' = (\tau_A,\tau_C,\tau_c)$.

    If we write $f$ (resp. $g$) the partial injection used in the computation of
    $\exc{R}$ (resp. $\exc{R'}$), then we can find a partial injection $h$
    such that the partial injection used in the computation of
    $\exc{\context{R}{I} \ACglue \context{R'}{I'}}$ is
    $f+g+h$. Moreover, we can decompose $h=i+i'$ in such a way that
    $h(\cord{I}{I'}\permplus\cord{I}{I'}^\star)h^\star
    = i \cord{I}{I'} i^\star \permplus i' \cord{I}{I'}^\star i'^\star$.

    The main part of the proposition amounts to proving that
    \begin{multline*}
    (\sigma_A\permplus\tau_A) \cexece{f+g+i+i'}
    (f\sigma_C f^\star \permplus g\tau_C g^\star \permplus 
    i \cord{I}{I'} i^\star \permplus i' \cord{I}{I'}^\star i'^\star)
    = \\
    (\sigma_A \cexece{f} f\sigma_C f^\star) \cexece{\cord{I}{I'}}
    (\tau_A \cexece{g} g \tau_C g^\star)
    \end{multline*}
    This equality can be deduced as in the proof of proposition~\ref{cor:assoc_comp_exec}.
    The fact that we have extra partial injections $f,g,i$ and $i'$ does not add
    any new difficulty.
\end{proof}

Therefore we can claim that
\begin{center}
    Interaction nets are the quotient of AC nets by $\exequiv$.
\end{center}

\section{Reduction by means of double pushout}
\subsection{Motivation}
In this section we briefly recall the double pushout approach of graph rewriting
and why we seek such kind of approach in our context.

We consider as rule of graph rewriting a diagram $R \leftarrow I \rightarrow S$
in $\cat{Graph}$, the category of graphs. The graph $I$ corresponds
to some sort of common interface between $R$ and $S$. As soon as we have
a morphism $R \rightarrow G$ we say that $G$ contains the redex of the rule, 
and we can construct in $\cat{Graph}$ a graph $G'$ such that we have a
pushout
    \begin{ctp}
        \matrix[matrix of math nodes,column sep=.3cm,row sep=.3cm]{
        |(Rr)| R &  & |(Ri)| I  \\
        & \text{po}   \\
        |(R)| G &  & |(S)| G' \\
        };
        \draw (Ri) edge[->] (Rr);

        \draw (Ri) edge[->]  (S);
        \draw (Rr) edge[->]  (R);

        \draw (S) edge[->]  (R);
    \end{ctp}
    A precise definition of pushout will be given later, but for now let us say that it
    corresponds to extracting $R$ from the graph $G$ while leaving the common part $I$.
    We can construct another pushout in the other direction, thus obtaining the
    diagram:
    \begin{ctp}
        \matrix[matrix of math nodes,column sep=.3cm,row sep=.3cm]{
        |(Rr)| R &  & |(Ri)| I & & |(Rp)| S \\
        & \text{po} & & \text{po} & \\
        |(R)| G &  & |(S)| G' & & |(T)| G_r \\
        };
        \draw (Ri) edge[->] (Rr);
        \draw (Ri) edge[->] (Rp);

        \draw (Ri) edge[->]  (S);
        \draw (Rr) edge[->]  (R);
        \draw (Rp) edge[->]  (T);

        \draw (S) edge[->]  (R);
        \draw (S) edge[->]  (T);
    \end{ctp}
    The graph $G_r$ is then called the reduct of $G$ by the rule.
    It is constructed by taking the graph $G'$ and replacing by $S$ the part
    left empty by the removing of $R$ in $G$, and then applying some
    kind of gluing operation along the interface $I$.

    This approach, initiated in the seminal paper~\cite{Ehrigh73}, leads
    to a definition of graph reduction which is at the same time
    intuitive and algebraically rigorous. It is quite natural to
    try to define it for interaction nets. Indeed, cutting and gluing
    are explicit operations in our framework.

    Note that such kind of approach for interaction nets is defined in the 
    paper~\cite{Banach95}, but it relies on an embedding of interaction nets in
    hypergraphs followed by an embedding of hypergraphs in bipartite graphs.
    In our setting, we can directly state the approach while staying in the
    realm of interaction nets.

\subsection{Pushouts in \texorpdfstring{$\cat{IN}$}{cat(N)}}
Let $\IN$ be the category whose objects are interaction nets and morphisms
are morphisms of interaction nets.

In this section we write $R \xinjarrow{f} S$ to say that
$f$ is an almost injective morphism from $R$ to $S$.
We write $R \xisoarrow{f} S$ when $f$ is a bijection.

We recall here the definition of pushouts.
\begin{defi}
    Let $\cat{C}$ be a category. A commutative square
    \begin{itp}
        \matrix[matrix of math nodes,column sep=.3cm,row sep=.3cm]{
        & |(R)| R & \\
        |(S)| S & & |(Sp)| S' \\
        & |(T)| T \\
        };
        \draw (R) edge[->] node[above]{$f$} (S);
        \draw (R) edge[->] node[auto]{$f'$} (Sp);
        \draw (S) edge[->] node[below]{$g$} (T);
        \draw (Sp) edge[->] node[auto]{$g'$} (T);
    \end{itp}
    is called a pushout whenever for any other commutative square
    \begin{itp}
        \matrix[matrix of math nodes,column sep=.3cm,row sep=.3cm]{
        & |(R)| R & \\
        |(S)| S & & |(Sp)| S' \\
        & |(T)| T' \\
        };
        \draw (R) edge[->] node[above]{$f$} (S);
        \draw (R) edge[->] node[auto]{$f'$} (Sp);
        \draw (S) edge[->] node[below]{$h$} (T);
        \draw (Sp) edge[->] node[auto]{$h'$} (T);
    \end{itp}
    there exists a unique 
     $T \xrightarrow{u} T'$ such that
    $u g = h$ and $u g' = h'$.
    
    We write \emph{po} in the center of a square to state that it is
    a pushout.
%
\end{defi}

The following lemma asserts that pushouts are stable under iso of their
branches, \emph{i.e.} that we can replace every middle object of the pushout square with
an isomorphic one. It will be useful to replace almost injective morphisms
by inclusion maps.
\begin{lem}
    Let 
    \begin{itp}
        \matrix[matrix of math nodes,column sep=.3cm,row sep=.3cm]{
        & |(R)| R & \\
        |(S)| S &  \text{po} & |(Sp)| S' \\
        & |(T)| T \\
        };
        \draw (R) edge[->] node[above]{$f$} (S);
        \draw (R) edge[->] node[auto]{$f'$} (Sp);
        \draw (S) edge[->] node[below]{$g$} (T);
        \draw (Sp) edge[->] node[auto]{$g'$} (T);
    \end{itp}
    be a pushout square and $S \xisoarrow{} \widetilde{S}$ be an
    iso. We also have the following pushout
    \begin{itp}
        \matrix[matrix of math nodes,column sep=.3cm,row sep=.3cm]{
        & |(R)| R & \\
        |(S)| \widetilde{S} &  \text{po} & |(Sp)| S' \\
        & |(T)| T \\
        };
        \draw (R) edge[->] node[above]{$h$} (S);
        \draw (R) edge[->] node[auto]{$f'$} (Sp);
        \draw (S) edge[->]  node[below]{$k$} (T);
        \draw (Sp) edge[->] node[auto]{$g'$} (T);
    \end{itp} \qed
\end{lem}
\begin{lem}
    \label{pushout subnets}
    We have the pushout
    \begin{ctp}
    \matrix[matrix of math nodes,column sep=.3cm,row sep=.3cm]{
        & |(R)| R \\
        |(S)| R \glue{f} \underbar{S} &  \text{po} & |(Sp)| R \glue{g} \underbar{S}' \\
        & |(T)| R \glue{f + g} (\underbar{S} + \underbar{S}') \\
        };
        \draw (R) edge[white,->] node[black,sloped,xscale=-1]{$\subseteq$} (S);
        \draw (R) edge[white,->] node[black,sloped]{$\subseteq$} (Sp);
        \draw (S) edge[white,->] node[black,sloped]{$\subseteq$} (T);
        \draw (Sp) edge[white,->] node[black,sloped,xscale=-1]{$\subseteq$} (T);
    \end{ctp}
    whenever $S$ and $S'$ are disjoint and 
    $\dom(f) \cap \dom(g) = \emptyset$.
\end{lem}

\begin{proof}
    Let 
    \begin{ctp}
    \matrix[matrix of math nodes,column sep=.3cm,row sep=.3cm]{
        & |(R)| R \\
        |(S)| R \glue{f} \underbar{S} &  & |(Sp)| R \glue{g} \underbar{S}' \\
        & |(T)| T \\
        };
        \draw (R) edge[white,->] node[black,sloped,xscale=-1]{$\subseteq$} (S);
        \draw (R) edge[white,->] node[black,sloped]{$\subseteq$} (Sp);
        \draw (S) edge[->] node[below]{$h$} (T);
        \draw (Sp) edge[->] node[below]{$h'$} (T);
    \end{ctp}
    be another commutative square. We will build a morphism $u$
    from $R \glue{f + g} (\underbar{S} + \underbar{S}')$
    to $T$.
    Let $p$ be a port belonging to $P(R) - \dom(f) - \dom(g)$ or
    $P(\underbar{S}) - \codom(f)$ 
    we just set $u(p) = h(p)$. Similarly we define $u(p) = h'(p)$
    when $p$ belongs to $P(\underbar{S}') - \codom(g)$.
    Now, if we take a $p \in \dom(f)$ we can properly define 
    its image $p'$ in $R \glue{f} \underbar{S}$. We set
    $u(p) = h(p') = h' (p)$. We proceed in the same way for a 
    $p \in \dom(g)$.

    By construction $u$ is unique and satisfies the required universal 
    property of pushouts.
\end{proof}

By using the two previous lemmas and the definition of almost injective morphisms, 
we get the following corollary.
\begin{cor}
    \label{theoreme:in sa}
    Let $S \xhookleftarrow{\alpha} R \xhookrightarrow{\beta} S'$ be a diagram in
    $\cat{IN}$ with $S,S'$ disjoint. By definition of almost injectivity we have
    $S = \widetilde{\alpha}(R) \glue{f} \underbar{S}$ and
    $S' = \widetilde{\beta}(R) \glue{g} \underbar{S}'$.

    If $\dom(f \widetilde{\alpha}) \cap \dom(g \widetilde{\beta}) = \emptyset$
    then we have the following diagram:
    \begin{ctp}
    \matrix[matrix of math nodes,column sep=.3cm,row sep=.7cm]{
        & |(R)| R \\
        |(S)| S &  \text{po} & |(Sp)| S' \\
        & |(T)| R \glue{f \widetilde{\alpha} + g \widetilde{\beta}} (\underbar{S} +
        \underbar{S}') \\
        };
        \draw (R) edge[left hook->] node[above left]{$\alpha$} (S);
        \draw (S) edge[left hook->] (T);
        \draw (R) edge[left hook->] node[above right]{$\beta$} (Sp);
        \draw (Sp) edge[left hook->] (T);
    \end{ctp}
\end{cor}

\begin{rem}
    The disjointness of $S$ and $S'$ in the previous lemma is not
    mandatory as pushouts are only defined up to isomorphism.
\end{rem}

\begin{lem}[Complement]
    \label{lemme:contexte}
    If we have
    \begin{itp}
        \matrix[matrix of math nodes,column sep=.3cm,row sep=.3cm]{
        & |(R)| R & \\
        |(S)| S &  \\
        & |(T)| T \\
        };
        \draw (R) edge[left hook->] node[above]{$\alpha$} (S);
        \draw (S) edge[left hook->] node[below]{$\beta$} (T);
    \end{itp}
    then there exists $S'$ and
    $R \xhookrightarrow{\alpha'} S'$ such that
    \begin{itp}
        \matrix[matrix of math nodes,column sep=.3cm,row sep=.3cm]{
        & |(R)| R \\
        |(S)| S &  \text{po} & |(Sp)| S' \\
        & |(T)| T \\
        };
        \draw (R) edge[left hook->] node[above]{$\alpha$} (S);
        \draw (S) edge[left hook->] node[below]{$\beta$} (T);
        \draw (R) edge[left hook->] node[above right]{$\alpha'$} (Sp);
        \draw (Sp) edge[white,->] node[black,sloped,xscale=-1]{$\subseteq$} (T);
    \end{itp}
\end{lem}
\begin{proof}
    First, we show that we only need to prove the result when all arrows
    are inclusion maps. Indeed, by applying the definition of almost injectivity
    we get the following commutative diagram:
    \begin{ctp}
        \matrix[matrix of math nodes,column sep=.3cm,row sep=.5cm]{
        & & &|(fR)| \bij{\alpha}(R) & & & |(gfR)| \bij{\beta \alpha}(R) \\
        |(R)| R & & & |(S)| S && |(gS)| \bij{\beta}(S) \\
        & & & &  & &  |(T)| T \\
        };
        \draw (R) edge[left hook->] node[below]{$\alpha$} (S);
        \draw (S) edge[bend right,left hook->] node[below]{$\beta$} (T);
        \draw (R) edge[->] node[above]{$\bij{\alpha}$} 
                           node[below,sloped]{$\sim$} (fR);
        \draw (fR) edge[->] node[above]{$\bij{\beta}$} 
                            node[below,sloped]{$\sim$} (gfR);
        \draw (S) edge[->] node[above]{$\bij{\beta}$} 
                           node[below,sloped]{$\sim$} (gS);
        \draw (gfR) edge[white,->] node[black,sloped,xscale=-1]
            {$\subseteq$} (gS);
        \draw (S) edge[white,->] 
            node[black,sloped,xscale=-1]{$\subseteq$} (fR);
        \draw (gS) edge[white,->] node[black,sloped]{$\subseteq$} (T);
    \end{ctp}
    If we could complete it with a pushout on the right, as in
    \begin{ctp}
        \matrix[matrix of math nodes,column sep=.3cm,row sep=.5cm]{
        & & &|(fR)| \bij{\alpha}(R) & & & |(gfR)| \bij{\beta \alpha}(R) \\
        |(R)| R & & & |(S)| S && |(gS)| \bij{\beta}(S) & \text{po} & 
        |(Sp)| S'\\
        & & & &  & &  |(T)| T \\
        };
        \draw (R) edge[left hook->] node[below]{$\alpha$} (S);
        \draw (S) edge[bend right,left hook->] node[below]{$\beta$} (T);
        \draw (R) edge[->] node[above]{$\bij{\alpha}$} 
                           node[below,sloped]{$\sim$} (fR);
        \draw (fR) edge[->] node[above]{$\bij{\beta}$} 
                            node[below,sloped]{$\sim$} (gfR);
        \draw (S) edge[->] node[above]{$\bij{\beta}$} 
                           node[below,sloped]{$\sim$} (gS);
        \draw (gfR) edge[white,->] node[black,sloped,xscale=-1]
            {$\subseteq$} (gS);
        \draw (S) edge[white,->] 
            node[black,sloped,xscale=-1]{$\subseteq$} (fR);
        \draw (gS) edge[white,->] node[black,sloped]{$\subseteq$} (T);
        \draw (gfR) edge[white,->] node[black,sloped]{$\subseteq$} (Sp);
        \draw (Sp) edge[white,->] node[black,sloped,xscale=-1]{$\subseteq$} (T);
    \end{ctp}
    we would get the main pushout.
  
    So, let us prove it in the case where
    $R \subseteq S \subseteq T$.
    By definition, we have $S = \underbar{R} \glue{f} R$
    and $T = S \glue{g} \underbar{S}$. Thus, we have
    $T = (\underbar{R} \glue{f} R) \glue{g} \underbar{S}$.
    By corollary~\ref{cor:associativity_gluing}, there exists
    $f_1$ and $f_2$ such that
    $T = \underbar{R} \glue{f_1} (R \glue{f_2+g} \underbar{S})$.
    We set $S' = R \glue{f_2+g} \underbar{S}$. 
    We can conclude by applying lemma~\ref{pushout subnets}. 
\end{proof}
\subsection{Generalized reduction}

\begin{defi}
    Let 
    $R_r \xhookleftarrow{\alpha_r} R_i \xhookrightarrow{\alpha_p} R_p$ be a diagram in $\IN$.
    By definition of almost injectivity we have
    $R_r = \widetilde{\alpha_r}(R_i) \glue{f_r} \underbar{R}_r$
    and $R_p = \widetilde{\alpha_p}(R_i) \glue{f_p} \underbar{R}_p$.
    
    We say that this diagram is a \emph{generalized rule} when
    $\dom(f_r \widetilde{\alpha_r}) = \dom(f_p \widetilde{\alpha_p})$.
\end{defi}

\begin{thm}
    If $R_r \xhookleftarrow{\alpha_r} R_i \xhookrightarrow{\alpha_p} S_p$ is
    a generalized rule and we have a morphism
    $R_r \xhookrightarrow{\beta} R$ 
    then we can do the following completion
    \begin{ctp}
        \matrix[matrix of math nodes,column sep=.3cm,row sep=.3cm]{
        |(Rr)| R_r &  & |(Ri)| R_i & & |(Rp)| R_p \\
        & \text{po} & & \text{po} & \\
        |(R)| R &  & |(S)| S & & |(T)| T \\
        };
        \draw (Ri) edge[left hook->] node[above]{$\alpha_r$} (Rr);
        \draw (Ri) edge[left hook->] node[above]{$\alpha_p$} (Rp);

        \draw (S) edge[white] node[black,sloped,xscale=-1]{$\subseteq$} (R);
        \draw (Ri) edge[left hook->]  (S);
        \draw (Rr) edge[left hook->] node[left]{$\beta$} (R);
        \draw (Rp) edge[left hook->]  (T);

        \draw (S) edge[left hook->]  (T);
    \end{ctp}
    $T$ is called the \emph{reduct of $R$ by the generalized rule}.
\end{thm}

\begin{proof}
    The proof is just a chaining of the two lemmas~\ref{lemme:contexte}
    and~\ref{theoreme:in sa}.  The condition of equality of domain
    in the definition of generalized rule ensures that
    the domain of the gluing function in $R_i \xinjarrow{} S$, being
    disjoint from the domain of the gluing function in $R_i 
    \xinjarrow{} R_r$ is also disjoint from the gluing function
    in $R_i \xinjarrow{} R_p$. Thus, the lemma~\ref{theoreme:in sa}
    is applicable.
\end{proof}

\begin{prop}
    This reduction is a generalization of the one defined in section~\ref{sec:dynamics}.
\end{prop}
\begin{proof}
    Indeed let $(\context{R_r}{I_r},\context{R_p}{I_p})$ be an interaction rule and
    set $I_r = d_1 > \dots > d_m$. We define a net
    $$R_i = \left( (d_1\ f_1) \dots (d_m\ f_m), 0 \right)$$
    with $m$ new free ports $f_i$.

    We directly have an inclusion $R_i \subseteq R_r = R_i \glue{f_r} \underbar{R}_r$ 
    and by definition of
    an interaction rule, we have a bijection between $I_r$ and $I_p$ which can be
    lifted to an almost injective morphism $R_i \xinjarrow{\alpha_p} R_p 
    = \widetilde{\alpha_p}(R_i) \glue{f_p} \underbar{R}_p$. 

    The diagram $R_r \supseteq R_i \xinjarrow{\alpha_p} R_p$ is a generalized rule
    as $\dom(f_r) = \dom(f_p \widetilde{\alpha_p}) = \{ f_1, \dots, f_m \}$.

    Now let $R_r \xinjarrow{\beta} R$ be an almost injective morphism, we have
    $R = \widetilde{\beta}(R_r) \glue{g} \underbar{R}$. We are going to 
    consider $\underbar{R}$ and $\underbar{R}_p$ disjoint, if it is not
    the case we just need to add an explicit renaming to the following
    computations. By construction,
    we get $S = \widetilde{\beta}(R_i) \glue{g'} \underbar{R}$, where $g'$ is
    the restriction of $g$ to $\widetilde{\beta}(R_i)$, and we have
    $T = R_i \glue{g' \widetilde{\beta} + f_p \widetilde{\alpha_p}}
    (\underbar{R} + \underbar{R}_p)$ which is the result of the previously defined
    reduction.
\end{proof}

\section{Implementation}
\subsection{Introduction}
We detail here part of our implementation in OCaml of an interaction net tool.
This implementation follows closely the mathematical definitions given
earlier. By doing so we hope that we make apparent the idea that this 
framework, even though involving mathematical objects, can be seen as 
a natural syntax for implementing interaction nets.

A self-contained net reducer has been extracted from our implementation and
is presented in Appendix~\ref{sec:code}. For the sake of briefness we
have removed from this code subroutines involving renaming of net.

\subsection{Data structures}
The easiest way to represent partial permutation is to
define them as their list of orbits. The fact that orbits
are disjoint and make sense will in fact be ensured by the
validity of our operations. 

We define two types
\begin{lstlisting}
type 'a lorbit = { cycle : int list; label : 'a } 
type 'a lperm = 'a lorbit list
\end{lstlisting}
for representing labelled permutation, and we
only need to set a dummy label to represent an unlabelled permutation.

Therefore, the type for representing a net is
\begin{lstlisting}
type cell_label = { symbol : symbol; pal : int }
type net = { cells : cell_label lperm; wires : unit lperm }
\end{lstlisting}

Following the previous definitions, we define interface, context
and rule
\begin{lstlisting}
type interface = int list 
type context =  net * interface 
type rule = { symbols : symbol * symbol; pattern : context }
\end{lstlisting}

\subsection{Algorithms}
To have a full implementation we need to be able to find when a
reduction rule could be applied, and then to apply it.
Nevertheless the only changing part between this framework and
the usual one is the use of $\Ex$-composition to define the reduction.

We recall here the standard procedure for reducing nets, next to each step we
give the corresponding functions in the code found in Appendix~\ref{sec:code}:
\begin{enumerate}[(1)]
\item Extract the list of active wires, \emph{i.e.} wires linking two principal ports\hfill\break [\verb+net_get_active_wires+]
\item Filter out the active wires corresponding to a rule redex\hfill\break [\verb+net_appliable_rules+]
\item For one of these matches, cut out the redex and replace it with the rule pattern\hfill\break [\verb+net_remove_cell+,\verb+net_remove_wire+,\verb+net_apply_rule+]
\end{enumerate}

\noindent The main difference here, is that our replacement of the pattern relies on 
a net gluing [\verb+net_glue+], which in turns relies on an $\Ex$-composition
[\verb+perm_excomp+].

\begin{algorithm}
    \caption{Computation of $\sigma_w \cexece{f} \tau_w$
    for $\sigma_w,\tau_w$ being $w$-permutations}
    \label{excomp compute}
\begin{algorithmic}
\STATE $orbits = \sigma_w$ + $\tau_w$
\FOR{$p \in \dom(f)$}
    \STATE $p' = f(p)$
    \STATE{ $w = $ orbit containing $p$ in $orbits$ }
    \STATE{ $w' = $ orbit containing $p'$ in $orbits$ }
    \STATE $orbits = orbits - [w,w']$
    \IF{$w = w'$}
        \STATE $orbits = [min(p,p')] :: orbits$
    \ELSE \STATE $(p,q) = w$
            \STATE $(p',q') = w'$
            \STATE $orbits = [q,q'] :: orbits$
    \ENDIF
\ENDFOR
\RETURN{$orbits$}
\end{algorithmic}
\end{algorithm}

A method for computing $\sigma_w \cexece{f} \tau_w$ can be found
in Algorithm~1. This algorithm amounts to concatenation
of orbits from $\sigma_w$ and $\tau_w$ by removing ports that are part of
$\dom(f) \cup \codom(f)$. If we consider that 
every operations used on permutations are linear, as it the case
with lists, its complexity is in $\mathcal{O}(|\dom(f)| 
(|\dom(\sigma_w)| + |\dom(\tau_w)|))$.  Note that in most cases
$|\dom(f)|$ is small compared to $|\dom(\sigma_w)| + |\dom(\tau_w)|$
because of the local aspect of reduction rules in interaction
nets.

\subsection{Extensions}
Our full interaction net tool\footnote{available in a preliminary version at the address
\url{http://marc.de-falco.fr/mlint}} deals with some common extensions of interaction nets.

To be able to handle sharing graphs, in the Abadi, Gonthier and Levy flavour~\cite{AbadiGonthierLevy92b} 
we need to add parameters to cells, the so-called \emph{levels}.
These parameters are both used to guard the applicability of a rule and add dependencies on the
redex parameters inside the rule pattern. Thus, we extend the previous types with
\begin{lstlisting}
type 'a cell_label = { symbol:symbol; pal:int; parameter:'a} 
type 'a rule = { symbols : symbol * symbol;
                 pattern : 'a * 'a -> 'a context;
                 guard : 'a * 'a -> bool }
\end{lstlisting}

Another common extension is found in differential interaction nets, presented in~\cite{EhrhardRegnier05b},
which handles not only
nets but formal sum of nets. Concerning the rules it amounts to multiple patterns, therefore
we only need to adapt the previous type of pattern to
\begin{lstlisting}
                 pattern : 'a * 'a -> 'a context list
\end{lstlisting}

We would like to emphasise on the fact that these extensions do not imply complex changes to the
code presented in Appendix~\ref{sec:code}. Indeed, our framework presented here for vanilla
interaction net is quite flexible and it could serve as a basis for a rigorous study of 
extensions of interaction nets.

\section*{Conclusion}
Throughout this paper we have developed a syntactical framework for dealing
with interaction nets while still being rigorous. Some specific extensions of
this framework -- for example the definition of paths in nets, their reduction
and its strong confluence -- can be found in~\cite{deFalcoThesis09}.

At this point, it is quite natural to ask about semantics. So far no general 
notion of denotational semantics for interaction nets can be found in the literature. The
closest examples are either based on geometry of interaction~\cite{Lafont97,deFalco08}
or experiments~\cite{Mazza07}, and all treat of specific cases (interaction combinators or
differential interaction nets). 
Building on this framework, the author has a proposal which will be presented in a further paper.

\section*{Acknowledgements}
The author would like to thank Laurent Regnier and the anonymous
referees of both versions of this paper for their insightful comments.

\bibliographystyle{alpha}
\bibliography{biblio}

\begin{thebibliography}{AJM94}

\bibitem[Abr05]{Abramsky05}
S.~Abramsky.
\newblock Abstract scalars, loops, and free traced and strongly compact closed
  categories.
\newblock In {\em Algebra and Coalgebra in Computer Science}, volume 3629 of
  {\em Lecture Notes in Computer Science}, pages 1--29. Springer, 2005.

\bibitem[AGL92]{AbadiGonthierLevy92b}
M.~Abadi, G.~Gonthier, and J-J. L{\'e}vy.
\newblock The geometry of optimal lambda reduction.
\newblock In {\em Proceedings of the 19$^{th}$ Annual ACM Symposium on
  Principles of Programming Languages}, pages 15--26. Association for Computing
  Machinery, ACM Press, 1992.

\bibitem[AJ92]{AbramskyJagadeesan92}
S.~Abramsky and R.~Jagadeesan.
\newblock {New foundations for the geometry of interaction}.
\newblock In {\em Proceedings of the 7$^{th}$ Symposium on Logic in Computer
  Science}, pages 211--222, Santa Cruz, 1992. IEEE Computer Society Press.

\bibitem[AJM94]{AbramskyJagadeesanMalacaria94}
S.~Abramsky, R.~Jagadeesan, and P.~Malacaria.
\newblock Full abstraction for {PCF} (extended abstract).
\newblock In Masami Hagiya and {John C.} Mitchell, editors, {\em Theoretical
  Aspects of Computer Software. International Symposium TACS'94}, number 789 in
  Lecture Notes in Computer Science, pages 1--15, Sendai, Japan, April 1994.
  Springer-Verlag.

\bibitem[Ban95]{Banach95}
R.~Banach.
\newblock {The algebraic theory of interaction nets}.
\newblock {\em Department of Computer Science, University of Manchester,
  Technical Report MUCS-95-7-2}, 1995.

\bibitem[dF08]{deFalco08}
M.~de~Falco.
\newblock The geometry of interaction of differential interaction nets.
\newblock In {\em Proceedings of the 23$^{th}$ Symposium on Logic in Computer
  Science}, Pittsburgh, 2008. IEEE Computer Society Press.

\bibitem[dF09]{deFalcoThesis09}
M.~de~Falco.
\newblock {\em Géométrie de l'Interaction et Réseaux Différentiels}.
\newblock Th\`ese de doctorat, Universit\'e Aix-Marseille~2, 2009.

\bibitem[DR95]{DanosRegnier95}
V.~Danos and L.~Regnier.
\newblock Proof-nets and the {Hilbert} space.
\newblock In Jean-Yves Girard, Yves Lafont, and Laurent Regnier, editors, {\em
  Advances in Linear Logic}, volume 222 of {\em London Mathematical Society
  Lecture Note Series}. Cambridge University Press, 1995.

\bibitem[EPS73]{Ehrigh73}
H.~Ehrig, M.~Pfender, and H.~J. Schneider.
\newblock {Graph-grammars: an algebraic approach}.
\newblock In {\em IEEE Conference Record of 14th Annual Symposium on Switching
  and Automata Theory, 1973. SWAT'08}, pages 167--180, 1973.

\bibitem[ER05]{EhrhardRegnier05b}
T.~Ehrhard and L.~Regnier.
\newblock Differential interaction nets.
\newblock In {\em Workshop on Logic, Language, Information and Computation
  (WoLLIC), invited paper}, volume 123 of {\em Electronic Notes in Theoretical
  Computer Science}. Elsevier, 2005.

\bibitem[FM99]{FernandezMackie99}
M.~Fernandez and I.~Mackie.
\newblock {A calculus for interaction nets}.
\newblock {\em Lecture Notes in Computer Science}, 1702:170--187, 1999.

\bibitem[Gir87]{Girard87c}
J-Y. Girard.
\newblock Multiplicatives.
\newblock In Lolli, editor, {\em Logic and Computer Science : New Trends and
  Applications}, pages 11--34, Torino, 1987. {Universit\`a di Torino}.
\newblock Rendiconti del seminario matematico dell'universit{\`a} e politecnico
  di Torino, special issue 1987.

\bibitem[Gir89]{Girard89}
J-Y. Girard.
\newblock Geometry of interaction {I}: an interpretation of system ${F}$.
\newblock In Valentini Ferro, Bonotto and Zanardo, editors, {\em Proceedings of
  the Logic Colloquium 88}, pages 221--260, Padova, 1989. North-Holland.

\bibitem[HO00]{HylandOng94}
M.~Hyland and L.~Ong.
\newblock On full abstraction for {PCF}.
\newblock {\em Information and Computation}, 163:285--408, 2000.

\bibitem[HS06]{HaghverdiScott06}
E.~Haghverdi and P.~Scott.
\newblock {A categorical model for the geometry of interaction}.
\newblock {\em Theoretical Computer Science}, 350(2-3):252--274, 2006.

\bibitem[JSV96]{JoyalStreetVerity96}
A.~Joyal, R.~Street, and D.~Verity.
\newblock {Traced monoidal categories, Math}.
\newblock In {\em Proc. Comb. Phil. Soc}, volume 119, pages 447--468, 1996.

\bibitem[KL80]{KellyLaplaza80}
G.~M. Kelly and M.~L. Laplaza.
\newblock {Coherence for compact closed categories}.
\newblock {\em Journal of Pure and Applied Algebra}, 19:193--213, 1980.

\bibitem[Laf90]{Lafont90}
Y.~Lafont.
\newblock Interaction nets.
\newblock In {\em Proceedings of the 17$^{th}$ Annual ACM Symposium on
  Principles of Programming Languages}, pages 95--108, San Francisco, 1990.
  Association for Computing Machinery, ACM Press.

\bibitem[Laf97]{Lafont97}
Y.~Lafont.
\newblock {Interaction Combinators}.
\newblock {\em Information and Computation}, 137(1):69--101, 1997.

\bibitem[Laf03]{Lafont03}
Y.~Lafont.
\newblock {Towards an Algebraic Theory of Boolean Circuits}.
\newblock {\em Journal of Pure and Applied Algebra}, 184(2-3):257--310, 2003.

\bibitem[Lip03]{Lippi03}
Sylvain Lippi.
\newblock {Encoding left reduction in the $\lambda$-calculus with interaction
  nets}.
\newblock {\em Mathematical Structures in Computer Science}, 12(06):797--822,
  2003.

\bibitem[Mac98]{Mackie98a}
I.~Mackie.
\newblock {YALE: yet another lambda evaluator based on interaction nets}.
\newblock In {\em Proceedings of the third ACM SIGPLAN international conference
  on Functional programming}, pages 117--128. ACM New York, NY, USA, 1998.

\bibitem[Maz07]{Mazza07}
D.~Mazza.
\newblock {A denotational semantics for the symmetric interaction combinators}.
\newblock {\em Mathematical Structures in Computer Science}, 17(03):527--562,
  2007.

\bibitem[MP98]{Mackie98b}
I.~Mackie and J.~S. Pinto.
\newblock {Compiling the Lambda Calculus into Interaction Combinators}.
\newblock In {\em Logical Abstract Machines workshop}, 1998.

\bibitem[Pin00]{Pinto00}
J.~S. Pinto.
\newblock {Sequential and concurrent abstract machines for interaction nets}.
\newblock In {\em FOSSACS '00: Proceedings of the Third International
  Conference on Foundations of Software Science and Computation Structures},
  pages 267--282. Springer-Verlag London, UK, 2000.

\bibitem[Vau07]{Vaux07}
L.~Vaux.
\newblock {\em Lambda-calcul différentiel et logique classique}.
\newblock Th\`ese de doctorat, Universit\'e de la Méditerranée, 2007.

\end{thebibliography}

\appendix
\section{A lightweight interaction net reducer in OCaml}
\label{sec:code}
\lstset{language=caml,
    basicstyle=\footnotesize
    }

\begin{lstlisting}
type symbol = string
type 'a lorbit = { cycle : int list; label : 'a } 
type 'a lperm = 'a lorbit list
type cell_label = { symbol:symbol; pal:int } 
type net = { cells : cell_label lorbit list; wires : unit lperm } 
type port = FreePort of int | CellPort of cell_label lorbit * int
type interface = int list 
type context =  net * interface 
type rule = { symbols : symbol * symbol; pattern : context }

(* Utility functions for handling orbits *)
let cycle l = (List.tl l)@[List.hd l]
let rec cycle_to e l =
    if List.hd l = e then l else cycle_to e (cycle l)
let rec index l p = match l with
                | [] -> raise Not_found
                | hd::tl -> if hd = p then 0 else 1+(index tl p) 
let rec filter_opt l = 
    match l with
    | [] -> []
    | None::l -> filter_opt l
    | Some a::l -> a::(filter_opt l)
let list_diff l1 l2 = List.filter (fun x ->  not (List.mem x l2)) l1

(* Get the orbit in p containing an element e 
 * This function returns a couple (o, p')
 * where o is the orbit and p' is the remaining permutation *)
let lperm_get_orbit_split e p  =
    let rec laux acc e p = 
        match p with
            [] -> raise Not_found
            | o::ol -> if List.mem e o.cycle 
                       then (o, acc@ol)
                       else laux (o::acc) e ol
    in laux [] e p

(* Get optionally the orbit containing e in p *)
let lperm_get_orbit e p = try 
        let o = fst (lperm_get_orbit_split e p) in Some o
    with Failure _ -> None


(* Get the element after e along its orbit in the permutation p *)
let lperm_next e p =
    let (o,_) = lperm_get_orbit_split e p in
        List.hd (cycle (cycle_to e o.cycle))

(* Get a new permutation p' from permutation p by 
 * fusing the orbit oa containing a and the orbit ob 
 * containing b. This is done by inserting ob inside oa in
 * such a way that p'(a) = b. In case oa = ob it adds a fixpoint. *)
let lperm_fuse_orbits (p:'a lperm) (a:int) (b:int) =
    let (oa, ola) = lperm_get_orbit_split a p in
    if List.mem b oa.cycle
        then { cycle=[min a b]; label=oa.label }::ola
        else let (ob, nol) = lperm_get_orbit_split b ola in
            { cycle=[lperm_next a p; lperm_next b p]; label=oa.label }::nol

(* Disjoint sum of permutations p1 and p2 *)
let lperm_sum p1 p2 = p1 @ p2

(* Compute the ex-composition of s and t along f by
 * fusing orbits in the union of s and t along f *)
let perm_excomp s t f =
    let rec fuse_orbits p l = match l with
        [] -> p
        | (a,b)::tl -> fuse_orbits (lperm_fuse_orbits p a b) tl in 
        fuse_orbits (s@t) f

(* Net gluing n1 <-f-> n2 *)
let net_glue n1 n2 f = { cells=lperm_sum n1.cells n2.cells; 
                         wires=perm_excomp n1.wires n2.wires f }
let net_sum n1 n2 = net_glue n1 n2 []
let coord i1 i2 = List.combine i1 i2
(* context gluing n1^i1 <-> n2^i2 *)
let context_glue (n1,i1) (n2,i2) = net_glue n1 n2 (coord i1 i2)

(* Discriminate a given port p in the net n *)
let net_get_port n p =
    match lperm_get_orbit p n.cells with
        | Some c -> let cycle = cycle_to c.label.pal c.cycle in
                CellPort(c, index cycle p)
        | None -> FreePort p

(* Predicate asserting the fact that w is an active wire in n *)
let net_wire_is_active n w = match w.cycle with
        | [ p1; p2 ] -> begin
            match (net_get_port n p1, net_get_port n p2) with
                | (CellPort(c1,_), CellPort(c2,_)) ->
                    p1 = c1.label.pal && p2 = c2.label.pal
                | _ -> false
            end
        | _ -> false
let net_get_active_wires n = List.filter (net_wire_is_active n) n.wires

(* Extract the cell containing the port p from the net n *)
let net_remove_cell n p =
    let (c,nc) = lperm_get_orbit_split p n.cells in
    let nw = List.filter 
        (fun x -> list_diff c.cycle x.cycle <> []) n.wires in
        (c, { cells=nc; wires=nw })
(* Extract the wire p1--p2 from the net n *)
let net_remove_wire n p1 p2 =
    let nw = List.filter 
        (fun w -> list_diff w.cycle [p1;p2] <> []) n.wires in
        { cells=n.cells; wires=nw }

(* Apply a rule in n by removing the redex containing the active 
 * wire p1--p2 and replacing it with ctx *)
let net_apply_rule n (p1,p2,ctx) = 
    let (c1,n1) = net_remove_cell n p1 in
    let (c2,n2) = net_remove_cell n1 p2 in
    let n3 = net_remove_wire n2 c1.label.pal c2.label.pal in
    let i = (list_diff c1.cycle [p1])@(list_diff c2.cycle [p2]) in
    context_glue (n3,i) ctx

(* Take a net n and a list of rules rl and return a sublist
 * of rules having a matching redex in n *)
let net_appliable_rules n rl =
    let law = net_get_active_wires n in
    let matching_rule (s1,s2) w = 
        match w.cycle with
        | [p1;p2] -> begin
            match (net_get_port n p1, net_get_port n p2) with
            | (CellPort(c1,_), CellPort(c2,_)) -> begin
                match (c1.label.symbol, c2.label.symbol) with
                | cs1, cs2 when cs1 = s1 && cs2 = s2 -> Some (p1, p2)
                | cs1, cs2 when cs1 = s2 && cs2 = s1 -> Some (p2, p1)
                | _ -> None
                end
            | _ -> None
            end
        | _ -> None in
    let res r = List.map 
            (fun (p1,p2) -> (p1,p2,r.pattern))
            (filter_opt (List.map (matching_rule r.symbols) law))
    in List.concat (List.map res rl)

(* Take a net n and a list of rules rl and return an optional
 * reduct *)
let net_reduce n rl = 
    let res = net_appliable_rules n rl in
    match res with
        | instance::_ -> Some (net_apply_rule n instance)
        | _ -> None
\end{lstlisting}

\end{document}